\documentclass[]{article}
\RequirePackage[OT1]{fontenc}
\RequirePackage{float}
\usepackage{amsmath}
\usepackage{amssymb}
\usepackage{bm}
\usepackage{setspace}
\usepackage[margin=1in]{geometry}
\usepackage[utf8]{inputenc}
\usepackage[english]{babel}
\usepackage{mathrsfs}
\usepackage{subfigure}
\usepackage{amsthm}
\usepackage{bbm}
\usepackage[authoryear]{natbib}
\usepackage{lipsum}
\usepackage{multirow}
\usepackage{graphicx}
\usepackage[colorlinks,linkcolor=black,citecolor=black,urlcolor=black,filecolor=blue,backref=page]{hyperref}
\usepackage[usenames,dvipsnames,svgnames,table]{xcolor}

\newtheorem{theorem}{Theorem}[section]
\newtheorem{corollary}[theorem]{Corollary}
\newtheorem{lemma}[theorem]{Lemma}

\newtheorem{remark}{Remark}

\newcommand{\bea}{\begin{eqnarray*}}
\newcommand{\eea}{\end{eqnarray*}}
\newcommand{\bean}{\begin{eqnarray}}
\newcommand{\eean}{\end{eqnarray}}
\newcommand{\lra}{\longrightarrow}

\newcommand{\bfX}{{\bf X}}
\newcommand{\bfY}{{\bf Y}}

\newcommand{\C}{{\rm Cov}}
\newcommand{\sg}{\Sigma}

\newcommand{\calS}{\mathcal{S}}

\newcommand{\bbP}{\mathbb{P}} 
\newcommand{\bbR}{\mathbb{R}}
\newcommand{\bbE}{\mathbb{E}}

\title{Consistent and scalable Bayesian joint variable and graph selection for disease diagnosis leveraging functional brain network 
}

\author{Xuan Cao\\Department of Mathematical Sciences, University of Cincinnati
	\and 
	Kyoungjae Lee\footnote{Corresponding author.} \\ Department of Statistics, Sungkyunkwan University}

\begin{document}
\allowdisplaybreaks
	\doublespacing
	\noindent
\maketitle

\begin{abstract}
	We consider the joint inference of regression coefficients and the inverse covariance matrix for covariates in high-dimensional probit regression, where the predictors are both relevant to the binary response and functionally related to one another.
	A hierarchical model with spike and slab priors over regression coefficients and the elements in the inverse covariance matrix is employed to simultaneously perform variable and graph selection.
	We establish joint selection consistency for both the variable and the underlying graph when the dimension of predictors is allowed to grow much larger than the sample size, which is the first theoretical result in the Bayesian literature. 
	A scalable Gibbs sampler is derived that performs better in high-dimensional simulation studies compared with other state-of-art methods. 
	We illustrate the practical impact and utilities of the proposed method via a functional MRI dataset, where both the regions of interest with altered functional activities and the underlying functional brain network are inferred and integrated together for stratifying disease risk.
\end{abstract}
	
\section{Introduction}
Analyzing high-dimensional data is becoming increasingly prevalent and challenging as technology advances facilitating the collection and storage of more extensive massive data.
When applying a generalized linear model (GLM) to such large-scale data, a large number of variables can easily cause an overfitting problem.
In this situation, variable selection is one of the most commonly used techniques to avoid overfitting.
Numerous frequentist methods on variable selection have been introduced ever since the appearance of Lasso \citep{tibshirani1996regression}, and many analogous Bayesian methods have also been proposed \citep{ishwaran2005,narisetty2014bayesian,rovckova2018spike}. 

On the other hand, understanding the complex relationships between variables high-dimensional datasets is also important, where inverse covariance matrices (or equivalently, precision matrices) are prevailingly exploited to capture the multivariate dependence.
This is often called a network structure between the variables.
A variety of work on algorithms and their theoretical considerations have emerged to investigate a network structure \citep{wainwright_2019}. 
One of key developments was the introduction of the neighborhood selection method \citep{meinshausen2006high}, which leverages the connection between the $(i,j)$th entry of the inverse covariance matrix $\Omega$ to the partial correlation between the $i$th and $j$th variable estimated through a penalized regression setup. 
Many other frequentist methods have been developed for sparse precision matrix estimation based on the neighborhood selection \citep{yuan2007ggm,Friedman2007glasso,Peng2009SPACE,Khare2015CONCORD}, and several Bayesian counterparts have been proposed in the literature \citep{dobra2011,wang2012,wang2015scaling}. However, a key challenge for these Bayesian approaches is their scalability to high-dimensional settings. 
To address this issue, recently, \citet{jalali2020bconcord} employed the regression-based generalized likelihood function in \citet{Khare2015CONCORD} combined the spike and slab priors over entries in $\Omega$.
They proposed a scalable Gibbs sampler that works well in high dimensions and runs comparably fast compared with the Graphical Lasso \citep{Friedman2007glasso}.

It is often of interest to jointly perform variable selection and discover the network structure among predictors.
This type of problems is of wide clinical applications in radiological and genomic studies.  
Magnetic resonance imaging (MRI) scans and genetic traits are typical examples where the mechanism for effect on an outcome, such as functional brain activities \citep{Functional:2012} or molecular phenotypes such as gene expression, proteomics, or metabolomics \citep{GeneExpressionNetwork:2007,MetabolomicsNetwork:2020}, often displays a coordinated change along a pathway.
In such cases, the impact of a single factor may not be apparent.
Specifically for radiological studies, recent progress in imaging analysis allows the development of a novel feature extraction method called radiomics which converts large amounts of medical imaging characteristics into high-dimensional mineable data pool to build a predictive and descriptive model. The method has been applied to the diagnosis of neuropsychiatric diseases such as autism, schizophrenia, and Alzheimer disease \citep{Feng2019AD,Salvatore2021}. These findings demonstrate the validity of these radiomic approaches in discovering discriminative features that can reveal pathological information. In such cases, the method of joint selection can incorporate and highlight the underlying brain network to improve the classification accuracy. 

Several frequentist and Bayesian methods have been proposed for joint inference on variables and graphs.
\citet{Li:Li:2008} and \citet{Li:Li:2010} investigated a graph-constrained regularization procedure as well as its theoretical properties in order to account for the neighborhood information of variables measured on a given graph.
\citet{Dobra:2009} estimated a network among relevant predictors by first performing a stochastic search to discover subsets of predictors, then using a Bayesian model averaging approach to estimate a dependency network. 
\citet{Liu:Chakraborty:2014}  developed a Bayesian method for regularized regression, which provides inference on the inter-relationship between variables by explicitly modeling through a graph Laplacian matrix. 
\citet{Joint:Network:2016} simultaneously inferred a sparse network among the predictors based on the block Gibbs sampler and performed variable selection using this network as guidance by incorporating it into a Markov random field (MRF) prior.

Despite recent advances in Bayesian methods for joint regression and covariance estimation, theory related to joint selection consistency is not well-understood. Some early attempts \citep{CL:Sinica} focused solely on linear regression models, where the predictors are linked through a directed graph with a known ordering.
To the best of our knowledge, joint variable and graph selection consistency in a high-dimensional GLM has not been investigated under either directed or undirected graphical models. 

In this paper, we consider a high-dimensional probit model with network-structured predictors via a Gaussian graphical model.
Our goal is to jointly perform variable and graph selection with theoretical guarantees, and to develop a scalable algorithm for joint inference in a high-dimensional regime. 
We fill the gap in the literature by establishing joint selection consistency of the proposed posterior distribution, which guarantees that the posterior probability assigned to the significant variables and the true graph tends to $1$ as we observe more data.
To perform joint selection, spike and slab priors, imposed on the regression coefficients and the precision matrix of predictors, are linked by an MRF prior.
Furthermore, for scalable inference, we adopt the regression-based generalized likelihood function \citep{Khare2015CONCORD} for the predictors. 
This enables the derivation of a scalable Gibbs sampler by making available the conditional posteriors for the entries of the precision matrix in closed form. 
We illustrate the practical impact and utilities of the proposed method via a functional MRI dataset, where both the regions of interest with altered functional activities and the underlying functional brain network are inferred and integrated together for disease diagnosis.

The rest of the paper is organized as follows. In Section \ref{sec:modelspecification}, we describe the generalized likelihood function for inverse covariance estimation and the spike and slab priors for sparsity recovery under a probit regression. 
Posterior computation algorithms are described in Section \ref{sec:computation}.
Theoretical results of the proposed posterior including joint variable and graph selection consistency are shown in Section \ref{sec:theory} with proofs provided in Section \ref{sec:proofs}.   
We show the performance of the proposed method and compare it with other competitors through simulation studies in Section \ref{sec:simulation}. 
In Section \ref{sec:PD}, a radiomic analysis is conducted for predicting Parkinson’s disease based on functional MRI (fMRI) data, and a discussion is given in Section \ref{sec:discussion}.

\section{Model Specification} \label{sec:modelspecification}
Consider a case-control study to identify the radiomic features that are network-structured and may contribute to the disease risk by comparing patients who have certain disease (the ``cases”) with subjects who do not have that disease but are otherwise similar (the ``controls”). In particular, for $i = 1, 2, \ldots, n$, let $Y_i \in \{0,1\}$ be the binary response variable indicating whether the $i$th subject has certain disease, and denote $X_i = (x_{i1}, x_{i2}, \ldots, x_{ip})^T \in \mathbb{R}^p$ as the covariate vector containing all the $p$ radiomic features for the $i$th subject. We consider the following probit model with covariates that obey a multivariate Gaussian distribution: for $1 \le i \le n,$
\begin{eqnarray} \label{model:spec}
& P\left(Y_i = 1 \mid X_i,\beta\right) = \Phi\big(X_i^T\beta\big), \label{model:spec1}
\\
&    X_i \mid \Omega  \overset{i.i.d.}{\sim} N_p \big(0, \Omega^{-1}\big), \label{model:spec2}
\end{eqnarray}  
where $\Phi(\cdot)$ is the cumulative distribution function of the standard normal distribution, $\beta$ is a $p \times 1$ vector of regression coefficients, and $\Omega = (\omega_{jk})$ denotes the $p \times p$ inverse covariance matrix. 
Our goal is to infer the regression coefficients $\beta$ and underlying network structure $\Omega$ simultaneously to identify all the significant features better.

\subsection{CONCORD generalized likelihood for predictors}

In the frequentist setting, one of the most popular methods to achieve a sparse estimate of $\Omega$ is the graphical lasso \citep{Friedman2007glasso,yuan2007ggm}, where  the objective function is composed of the negative Gaussian log-likelihood and an $\ell_1$-penalty term for the off-diagonal entries of the inverse covariance matrix over the space of positive definite matrices. 
This objective function is also proportional to  the posterior density of $\Omega$ under Laplace priors for the off-diagonal entries, leading to a Bayesian inference and analysis framework \citep{wang2012}. Note that the requirement on the positive definiteness of $\Omega$ translates to the expensive computational need of inverting $(p-1)\times(p-1)$ matrices in each iteration of both graphical lasso or Bayesian Markov Chain Monte Carlo (MCMC) algorithms. 

To mitigate this issue, \citet{Khare2015CONCORD} relaxed the parameter space of $\Omega$ from positive definite matrices to symmetric matrices with positive diagonal entries.
Note that it cannot be achieved under the graphical lasso framework due to the determinant of $\Omega$ in the likelihood function.
Let $S =  n^{-1} \sum_{i=1}^n X_i X_i^T$ denote the sample covariance matrix. 
They introduced the CONvex CORrelation selection methoD (CONCORD) generalized likelihood function, for a given $p\times p$ symmetric matrix $\Omega$,
\bean \label{likelihood_concord}
\mathcal L(\Omega) \,=\,
 \exp \bigg\{  n\sum_{j=1}^p \log \omega_{jj}  - \frac n 2 \mbox{tr}(\Omega^2S) \bigg\}
 \,=\,
\exp\bigg\{n\sum_{j=1}^p \log \omega_{jj} -  \frac 1 2 \sum_{j=1}^p\sum_{i=1}^n \Big(       \omega_{jj}x_{ij} + \sum_{k\neq j}\omega_{jk}x_{ik}\Big)^2 \bigg\},
\eean
which is motivated by the regression-based neighborhood selection method \citep{meinshausen2006high}.
The quadratic nature of the objective function \eqref{likelihood_concord} and the relaxation of the parameter space lead to an entire order of magnitude decrease in computational complexity compared to that required by graphical lasso-based approaches. 
Hereafter, we proceed with the CONCORD generalized likelihood \eqref{likelihood_concord} instead of the Gaussian likelihood corresponding to \eqref{model:spec2},
and show that asymptotic properties as well as the computational efficiency can be achieved under the Bayesian framework of joint inference.

\subsection{Spike and slab priors for graph selection}

The main goal of this paper is to simultaneously infer the sparsity pattern in both $\beta$ and $\Omega$. 
To facilitate this purpose, we first introduce the following spike and slab priors for every off-diagonal entry of $\Omega$,
\begin{equation} \label{off-diag}
\omega_{jk} \,\overset{ind}{\sim}\, (1-q)\delta_{0}(\omega_{jk}) + q N(0, 1/\lambda_{jk})   \quad \text{ for } 1 \le j < k \le p ,
\end{equation}
where $\delta_{0}(\cdot)$ denotes the point mass at 0, $\lambda_{jk}>0$ is the precision of slab part, and $q \in (0,1)$ is the prior inclusion probability. 
For the diagonal entries of $\Omega$, we assume
\begin{equation}\label{diag}
\omega_{jj} \,\overset{ind}{\sim}\,  {\rm Exp}(\lambda_j)  \quad \text{ for } 1\le j \le p  ,
\end{equation}
where $\lambda_j > 0$.
Let $\xi = (\omega_{jk}, 1 \le j< k < p)^T \in \bbR^{\binom p 2}$  and $\delta = ( \omega_{11}, \omega_{22}, \ldots, \omega_{pp} )^T \in \bbR^p$
be the collection of all the off-diagonal and diagonal entries of $\Omega$, respectively.
Let a symmetric matrix $G=\left(G_{jk}\right) \in \{0,1 \}^{p \times p}$ with zero diagonals represent the adjacency matrix corresponding to the precision matrix $\Omega$ where $G_{jk} = G_{kj} = 1$ if and only if $\omega_{jk} \neq 0$, and $G_{jk} = G_{kj} = 0$ otherwise. 
If we further restrict our analysis to only realistic models, i.e., precision matrices with nonzero entries no more than $R_1 > 0$, spike and slab priors \eqref{off-diag} can be alternatively represented as
\begin{eqnarray*}
&\xi \mid G \sim N_{|G|}(0, \Lambda_u), \\
&\pi( G )  \propto  q^{|G|}(1-q)^{\binom p 2 - |G|} I ( |G|  < R_1 )  ,
\end{eqnarray*}
where $|G| = \sum_{j=1}^{p-1} \sum_{k=j+1}^p  G_{jk}$ is the number of nonzero entries in the upper triangular part of $G$, $\Lambda$ is a diagonal matrix with diagonal entries $\{\lambda_{jk}, \, 1 \le j < k < p\}$, and $\Lambda_u$ is the sub-matrix of $\Lambda$ after removing the columns and rows corresponding to the zero indices in the upper triangular part of $G$ \citep{jalali2020bconcord}. 
In the above, $I(\cdot)$ stands for the indicator function.

\subsection{Incorporating graph structure for variable selection}

We denote a variable indicator $\gamma = \left\{\gamma_1, \gamma_2,  \ldots, \gamma_p\right\}$ such that $\gamma_j = 1$ if and only if $\beta_j \neq 0$, for $1 \le j \le p$. 
Let $  \beta_\gamma \in \bbR^{|\gamma|}$ be the vector formed by the active components in $\beta$ corresponding to model $\gamma$, where $|\gamma| = \sum_{j=1}^p \gamma_j$ is the number of nonzero entries in $\gamma$.
For any matrix $A \in \bbR^{q \times p}$ with $p$ columns, let $A_\gamma \in \bbR^{q \times |\gamma|}$ represent the submatrix formed from the columns of $A$ corresponding to the nonzero indices in model $\gamma$. 

For variable selection, we consider the following hierarchical prior over $\beta$:
\begin{eqnarray}
& \beta_\gamma \mid   \gamma \sim N_{|\gamma|} \left(0, \tau^2 I_{|\gamma|}\right), \label{beta}
\\
& \pi(  \gamma \mid G) \propto \exp\left(-a |\gamma| +   b \gamma^T {G} \gamma \right)I(|\gamma| <  R_2) ,   \label{MRF}
\end{eqnarray}
for some constants $a >0$, $b \ge 0$ and a positive integer $0\le R_2 \le p$.
Prior \eqref{beta} can be seen as a collection of slabs of spike and slab priors for regression coefficients \citep{narisetty2014bayesian,yang:2016}, where $\tau^2$ is the variance of the slab. 
Prior \eqref{MRF} is called an MRF prior on the variable indicator $\gamma$.
It encourages the inclusion of variables connected to other variables through the adjacency matrix $G$. 
MRF priors have been used in the variable selection literature including \citet{Joint:Network:2016, Li:Zhang:2010} and \citet{Stingo:Vannucci:2010}. 
Note that the hyperparameter $a$ in \eqref{MRF} corresponds to a penalty for large models, and $b$ determines how strongly an adjacency matrix $G$ affects inclusion probabilities of variables.
We can jointly infer a variable indicator $\gamma$ and an adjacency matrix $G$ by considering $b >0$, whereas $b=0$ leads to a separate inference of $\gamma$ and $G$.

\section{Posterior Computation} \label{sec:computation}
Model \eqref{model:spec1} is equivalent to letting $Y_i = I(Z_i \ge 0)$,
where $Z_i$ is an underlying continuous variable that has a normal distribution with mean $X_i^T\beta$ and variance 1. As we shall demonstrate subsequently, one can exploit this reparameterization to formulate a Gibbs sampler for posterior inference.
Let $Z = (Z_1, Z_2, \ldots, Z_n)^T$.
Combining this with the CONCORD generalized likelihood \eqref{likelihood_concord} and priors \eqref{off-diag}--\eqref{MRF}, the full posterior of $Z, \beta, \gamma, \Omega$ and $G$ is given by 
\bea
\pi( Z, \beta, \gamma, \Omega, G \mid Y,X  ) 
&\propto &  \exp \Big\{-\frac 1 {2}(Z - X_\gamma\beta_\gamma)^T(Z - X_\gamma\beta_\gamma)\Big\} \prod_{i = 1}^{n} \big\{Y_iI(Z_i \ge 0) + (1-Y_i)I(Z_i < 0)\big\} \\
&&\times\,\,  \pi(\gamma \mid G) \prod_{j: \gamma_j=1}(2\pi\tau^2)^{-1/2}\exp\left\{-\beta_{j}^2/(2\tau^2)\right\}\prod_{j: \gamma_j=0} I (\beta_{j} = 0)  \\
&&\times\,\, \exp\bigg\{n\sum_{j=1}^p \log \omega_{jj} -  \frac 1 2 \sum_{j=1}^p\sum_{i=1}^n \Big(       \omega_{jj}x_{ij} + \sum_{k\neq j}\omega_{jk}x_{ik}\Big)^2 - \sum_{j=1}^p \lambda_j\omega_{jj} \bigg\} \\
&&\times\,\,  \pi(G)  \mathop{\prod\prod}_{1\le j<k\le p} \Big\{(1-G_{jk})\delta_0(\omega_{jk}) + G_{jk}\lambda_{jk}^{1/2}/(2\pi)^{1/2}\exp\big(-\lambda_{jk}\omega_{jk}^2/2\big)\Big\} .
\eea
For the selection of shrinkage parameters $\lambda_{jk}$ and $\lambda_{j}$, following \cite{Park2008} and \cite{jalali2020bconcord}, we assign independent gamma prior distributions on each shrinkage parameter, i.e., $\lambda_{jk} \sim \mbox{Gamma}(r,s)$ for $1\le j < k \le p$ and $\lambda_j \sim \mbox{Gamma}(r,s)$ for $1 \le j \le p$, where $r$ and $s$ are some fixed positive hyperparameters. 


\subsection{Gibbs sampler}
We suggest using the standard Gibbs sampling for posterior inference. 
In particular, when sampling the off-diagonal entries of $\Omega$ and $G$, we modify the entrywise Gibbs sampler proposed by \cite{jalali2020bconcord} due to the MRF prior. 
For any matrix $A =(a_{jk}) \in\bbR^{p\times p}$ and $1 \le j \le k \le p$, let $A_{-jk}$ denote all the upper triangular entries of $A$, including diagonals, except $a_{jk}$. 
For $1 \le j \le p$, let $\beta_{-j} \in \bbR^{p-1}$ and $X_{-j} \in \bbR^{n \times (p-1)}$ denote the $\beta$ vector without the $j$th predictor and the submatrix of $X$ corresponding to $\beta_{-j}$, respectively. 
Let $\tilde{X}_j \in \bbR^n$ be the $j$th column of $X$.
The above full posterior leads to the following Gibbs sampler.
\begin{itemize}
	\item For  $1 \le i \le n$, generate $Z_i$ via the following conditional distribution,
	\begin{equation*} \label{gibbs_Z}
	\pi(Z_i \mid   Y, X, \beta  ) \propto  \begin{cases}
	N(Z_i \mid X_i^T \beta, 1) \mathbbm{1}\left(Z_i > 0\right), \quad\mbox{if } Y_i = 1,\\
	N(Y_i \mid X_i^T \beta,1) \mathbbm{1}\left(Z_i < 0\right), \quad\mbox{if } Y_i = 0 .
	\end{cases} 
	\end{equation*}

	\item For $1 \le j \le p$,
	set $\gamma_j = 0$ if $|\gamma_{-j}| = R_2-1$.
	Otherwise, generate $\gamma_j$ from the conditional distribution,
	\begin{align*}
	\gamma_j \mid X, Z, G, \gamma_{-j}, \beta_{-j} \sim \mbox{Bernoulli}\Big(\frac{d_j}{1 + d_j}\Big)  ,
	\end{align*}
	where $d_j = (\sigma_{j}/\tau^2 )^{1/2}\exp\big\{-a + 2b\sum_{i \neq j} \gamma_i G_{ij} + \mu_j^2/ (2 \sigma_{j}) \big\}$,
	$\sigma_{j} = (  \tilde{X}_{j}^T \tilde{X}_{j} + \tau^{-2} )^{-1}$ and $\mu_{j} = \sigma_{j} \tilde{X}_{j}^T (Z - X_{-j}\beta_{-j})$.
	
	\item For $1 \le j \le p$, generate $\beta_{j}$ based on the following spike and slab distribution,
	\begin{align*}
	\beta_{j} \mid X, Z, G, \gamma_{j}, \beta_{-j} \sim (1 - \gamma_j)\delta_0  + \gamma_jN(\mu_{j}, \sigma_{j}).
	\end{align*}

	\item For $1 \le j < k \le p$, 
	set $G_{jk} =0 $ if $|G_{-jk}| = R_1-1$.
	Otherwise, generate $G_{jk}$ based on
	\begin{align*}
		G_{jk}  \mid  \Omega_{-jk}, \gamma, X \sim \mbox{Bernoulli} \Big(\frac{c_{jk}}{1+c_{jk}}\Big),
	\end{align*}	
	where $S =  n^{-1} \sum_{i=1}^n X_i X_i^T  = (s_{jk})$ and 
	\bea
	&&a_{jk} = s_{jj} + s_{kk} + \frac{\lambda_{jk}}n, \quad  b_{jk} = \sum_{k^\prime \neq k} \omega_{jk^\prime}s_{kk^\prime} + \sum_{j^\prime \neq j} \omega_{j^\prime k}s_{jj^\prime}, \\
	&&c_{jk} = \frac q {1-q} \Big(\frac{\lambda_{jk}}{na_{jk}}\Big)^{\frac 1 2}\exp\Big(\frac{nb_{jk}^2}{2a_{jk}} + 2b\gamma_j\gamma_k\Big).
	\eea
	
	\item For $1 \le j < k \le p$, generate $\omega_{jk}$ based on the following spike and slab distribution,
	\begin{align*} 
		\omega_{jk} \mid G_{jk},  \Omega_{-jk}, \gamma, X \sim (1 - G_{jk})\delta_0(\omega_{jk}) + G_{jk} N\Big(-\frac{b_{jk}}{a_{jk}},\frac 1 {na_{jk}}\Big).
	\end{align*}

	\item  For $1 \le j < k \le p$, the conditional distribution of $\lambda_{jk}$ is given by 
	\begin{align*}
		\lambda_{jk} \mid \Omega \sim \mbox{Gamma} \big(r + 1/2, \omega_{jk}^2/2 + s \big) .
	\end{align*}

	\item For $1 \le j \le p$, the conditional distribution of $\lambda_j$ is given by
	\begin{align} \label{gibbs_lambda_j}
	\lambda_{j} \mid \omega_{jj} \sim \mbox{Gamma} \big(r + 1, \omega_{jj} + s \big) .
	\end{align}

	\item For $1 \le j \le p$, the conditional distribution of $\omega_{jj}$ is 
	$\pi(\omega_{jj} \mid \Omega_{-jj}, X) \propto \omega_{jj}^n \exp \big\{-ns_{jj}\omega_{jj}^2/2 - \omega_{jj}(\lambda_{j}+nb_j) \big\}$,
	whose normalizing constant is intractable, where $b_j = \sum_{j^\prime \neq j} \omega_{jj^\prime} s_{jj^\prime}$.
	As suggested by \cite{jalali2020bconcord}, we set $\omega_{jj}$ as the unique mode of $\pi(\omega_{jj} \mid \Omega_{-jj}, X)$, 
	\begin{align} \label{gibbs_omega_j}
	\omega_{jj}^\star \,\,=\,\, \frac{-(\lambda_j + nb_j) + \sqrt{(\lambda_j + nb_j)^2 + 4n^2 s_{jj}}}{2ns_{jj}}.
	\end{align}	
\end{itemize}
When sampling $\gamma_j$ and $G_{jk}$, we are using the conditional posteriors after integrating out $\beta_{j}$ and $\omega_{jk}$ respectively, rather than using the full conditional posterior.
This is to ensure that the Markov chain will be irreducible and converge, where the same trick has been commonly used, for examples, in \cite{Yang:Naveen:2018} and \cite{Xu:Ghosh:2015}.

\begin{remark}
An extensive numerical study conducted by \citet{jalali2020bconcord} showed that $\pi(\omega_{jj} \mid \Omega_{-jj}, X)$ puts most of its mass around the mode \eqref{gibbs_omega_j}.
By using this fact, we simply approximate the nonstandard density using the degenerate distribution at the mode for fast inference.
Otherwise, one can employ a Metropolis-Hastings algorithm to obtain samples from $\pi(\omega_{jj} \mid \Omega_{-jj}, X)$.
For example, the uniform distribution ${\rm Unif}( \omega_{jj}^\star/2, \, 2 \omega_{jj}^\star )$ can be used as a Metropolis-Hastings kernel.
%
%
\end{remark}

\section{Theoretical Properties}\label{sec:theory}

For any positive sequences $a_n$ and $b_n$, we denote (i) $a_n \gg b_n$ if $a_n /b_n \lra \infty$ as $n\to\infty$,
(ii) $a_n = O(b_n)$ if there exists a constant $C>0$ such that $a_n / b_n \le C$, 
(iii) $a_n \sim b_n$ if $a_n = O(b_n)$ and $b_n = O(a_n)$, and
(iv) $a_n = o(b_n)$ if $a_n / b_n \lra 0$ as $n\to\infty$, 
For any $a = (a_1, a_2, \ldots, a_p)^T\in \bbR^p$, we denote vector norms by $\|a\|_1 = \sum_{j=1}^p |a_j|$, $\|a\|_2 = \big( \sum_{j=1}^p a_j^2 \big)^{1/2}$ and $\|a\|_{\max} = \max_{1\le j \le p} |a_j|$.

In this section, we investigate asymptotic theoretical properties of the proposed Bayesian joint variable and graph selection method.
We are interested in whether the joint posterior for the variable and graph is concentrated on each true value.
Let $\beta_0  = ( \beta_{0,j} ) \in \bbR^p$ be the true coefficient vector, and $\gamma_0 =(\gamma_{0,j}) \in \{ 0,1 \}^p$ be the binary vector indicating locations of nonzero entries in $\beta_0$, i.e., $\gamma_{0,j} = I( \beta_{0,j} \neq 0 )$ for $j=1,2, \ldots, p$.
Let $\Omega_0 = (\omega_{0,jk}) \in \bbR^{p\times p}$ be the true precision matrix of $X_i$, and $G_0 = ( G_{0,jk} )  \in \{0,1\}^{p \times p}$ be the corresponding adjacency matrix.
Based on these quantities, we assume that the true data-generating mechanism is $Y_i \mid X_i, \beta_0  \,\overset{ind}{\sim}\, {\rm Ber}( \Phi(X_i^T \beta_0) )$ with a random predictor vector $X_i$ such that $\C(X_i) = \Omega_0^{-1}$, for  $i=1,2, \ldots, n$.
The following assumptions were made in order to demonstrate the theoretical properties.
In the below, $\bbP_0$ and $\bbE_0$ denote the probability measure and expectation, respectively, under the true data-generating mechanism.

\noindent{\bf Condition (A1)} (Conditions on $n$ and $p$)
$p = p_n \ge n$ and $\log p =o( n )$ as $n\to\infty$.

\noindent{\bf Condition (A2)} (Conditions on the design matrix)
For $X_i \in \bbR^p$, $i=1,2, \ldots, p$, we assume the following:
\begin{enumerate}
	\item[(i)] (sub-gaussianity) There exists a constant $C>0$ such that $\bbE_0 \exp ( \alpha^T X_i ) \le \exp ( C \|\alpha\|_2^2 )$ for all $\alpha \in\bbR^p$.
		
	\item[(ii)] (bounded eigenvalues) There exists a constant $0<\epsilon_0 <1$ such that
	$\epsilon_0 \le \lambda_{\min}(\Omega_0) \le \lambda_{\max}(\Omega_0)  \le \epsilon_0^{-1}$.
	
	\item[(iii)] (boundedness) $\bbP_0\big( \|X_i\|_{\max}  \le M \big) = 1$ for some constant $M>0$.
	
	\item[(iv)]  $(|G_0| + 1 ) ^2 \log p= o (n)$ and $\Omega_{0, \min} \equiv \min_{(j,k): G_{0,jk} = 1} \omega_{0,jk}^2 \gg  \{|G_0| \log p +( \log n) / 2\}  /n$.
	
\end{enumerate}

\noindent{\bf Condition (A3)} (Conditions on $\beta_0$)
$|\gamma_0| = O( 1)$, $\|\beta_0\|_1 = O(1)$ and  $\beta_{0,\min}^2 \equiv \min_{j \in \gamma_0 }  \beta_{0,j}^2 \ge C_{\beta_0} \log p /n$ for some constant $C_{\beta_0}>0$.

\noindent{\bf Condition (A4)} (Conditions on $\Omega_0$)
$(|G_0| + 1 ) ^2 \log p= o (n)$ and $\Omega_{0, \min} \equiv \min_{(j,k): G_{0,jk} = 1} \omega_{0,jk}^2 \gg  \{|G_0| \log p +( \log n) / 2\}  /n$.

\noindent{\bf Condition (A5)} (Conditions on the hyperparameter $q$)
$q  = p^{- C_q |G_0|}$, where $C_q = 16 (1 \vee c_0)^2 / (1 \wedge \epsilon_0)$, for some constant $c_0>0$ defined in Lemma S3 of \cite{jalali2020bconcord}.

\noindent{\bf Condition (A6)} (Conditions on the other hyperparameters)
For some constants $1/2 < d < 1$, $\delta>0$ and $C_a>0$, $R_1 = ( n / \log p)^{\frac{1}{2} }$,
$R_2 = (n / \log p)^{\frac{1-d}{2} }$, $\tau^2 \sim n^{-1} p^{2 + 2\delta}$, $a = C_a \log p$ and $b = o \big( (\log p /n)^{1-d}  \big)$.

Condition (A1) demonstrates the high-dimensional setting, where the number of variables $p$ is larger than the sample size $n$.
It allows $p$ to grow at a rate $\exp \{ o(n) \}$ as $n\to\infty$.
Similar conditions have been used in the literature including \cite{narisetty2014bayesian} and \cite{lee2021bayesian} to prove selection consistency of coefficient vector.

Condition (A2) shows the conditions for each row, $X_i$, of the random design matrix $X$.
The first condition implies that a linear combination of $X_i = (x_{i1}, x_{i2}, \ldots, x_{ip})^T$ has a sufficiently light tail satisfying sub-gaussianity.
The second condition requires that the eigenvalues of precision matrix $\Omega_0$ are bounded.
\cite{liu2019empirical} and \cite{cao2021joint} also used this condition for linear regression models with random design matrix.
The third condition requires each component of $X_i$ is bounded with probability $1$, where \cite{narisetty2019} adopted a similar condition for a deterministic design matrix.
By assuming these conditions for $X$, we can efficiently control the eigenvalues of $n^{-1} X_\gamma^T X_\gamma$ and the Hessian matrix of \eqref{model:spec1} for any reasonably large model $\gamma$, with large probability tend to $1$.
For example, condition (A2) holds if $X_i = \Omega_0^{-1/2} Z_i$, where $Z_i \overset{i.i.d.}{\sim} {\rm Unif}( [ -\sqrt{3} , \sqrt{3} ]^p )$ for $i=1,2,\ldots, n$ and $\|\Omega_0 \|_1 = O(1)$.
Here, $\|\cdot\|_1$ denotes the matrix $\ell_1$-norm.

Condition (A3) means that the true regression coefficient $\beta_0$ has finite numbers of nonzero entries and a bounded $\ell_1$-norm.
It holds that if we assume $\|\beta_0\|_{\max} = O(1)$.
For examples, \cite{johnson2012bayesian} and \cite{narisetty2014bayesian} assumed similar conditions.
Note that we still allow, as the number of variables increases, the magnitude of the smallest coefficient converge to zero at the rate of $\log p /n$.
This can describe a situation in which the importance of meaningful variables decreases as the number of variables grows.

Condition (A4) requires the number of nonzero off-diagonal entries in $\Omega_0$ is at most $O (\sqrt{n/ \log p})$.
\cite{banerjee2015bayesian}, \cite{xiang2015high} and \cite{lee2021bayesiandecomposable} used similar conditions for high-dimensional precision matrices.
Furthermore, condition (A4) allows the magnitude of the smallest nonzero off-diagonal elements of $\Omega_0$ converge to zero at the rate $(|G_0| \log p + \log n) / n$.
We adopt these conditions from \cite{jalali2020bconcord} to use their results.

Among conditions (A5) and (A6), $q  = p^{- C_q |G_0|}$ and $a = C_a \log p$ mean that the prior should impose a sufficient penalty to large $|G|$ and $|\gamma|$, respectively.
These are standard assumptions for Bayesian inference of high-dimensional precision matrix and regression vector.
For examples, see \cite{liu2019empirical}, \cite{jalali2020bconcord}, \cite{cao2019posterior} and \cite{martin2017empirical}.
The condition $\tau^2 \sim n^{-1} p^{2 + 2\delta}$ implies that the variance of slab part should be sufficiently large, where $\tau^2$ essentially plays a role as a penalty for large $|\gamma|$.
The other conditions, $R_1 = ( n / \log p)^{\frac{1}{2} }$ and $R_2 = (n / \log p)^{\frac{1-d}{2} }$ control the size of $|G|$ and $|\gamma|$, respectively, while $b = o \big( (\log p /n)^{1-d}  \big)$ controls the strength of $\gamma^T G \gamma$ term in $\pi(\gamma \mid G)$.
Similar conditions can be found in \cite{jalali2020bconcord} and \cite{cao2021joint}.

With these conditions at hand, we are now ready to state asymptotic properties of the posterior.
Theorem \ref{thm:gamma_cons} shows the proposed prior enjoys posterior ratio consistency of $\gamma$ given any $G$.
This implies that for any fixed $G$, the true variable indicator $\gamma_0$ is the mode of the conditional posterior $\pi(\gamma \mid G, Y, X)$ with probability tending to $1$.

\begin{theorem}[Posterior ratio consistency of $\gamma$]\label{thm:gamma_cons}
	Suppose conditions (A1)--(A3) and (A6) hold.
	Then, for any $\gamma \neq \gamma_0$ and $G$,
	\bea
	\frac{\pi(\gamma ,G \mid Y, X) }{\pi(\gamma_0 ,G \mid Y, X)}
	&\overset{P}{\lra}& 0  \quad \text{ as } n \to \infty . 
	\eea
\end{theorem}

To establish posterior ratio consistency of $G$ given $\gamma_0$., we assume the existence of accurate estimates of diagonal entries $\delta =( \omega_{11},  \omega_{22}, \ldots, \omega_{pp})$, say $\hat{\delta}=(\hat{\omega}_{11}, \hat{\omega}_{22}, \ldots, \hat{\omega}_{pp})$, satisfying $\| \delta- \hat{\delta}\|_{\max} = O \big( \sqrt{\log p /n} \big)$ with probability at least $1 - n^{-c}$ for any constant $c>0$.
The existence of these estimates have been commonly assumed for high-dimensional precision matrix estimation \citep{Peng2009SPACE,Khare2015CONCORD}; for example, Proposition 1 in \cite{Peng2009SPACE} provides one way to obtain such estimates of $\delta$.
Because our main focus is selection of $\gamma$ and $G$, not the estimation of $\Omega$, we will work with the conditional posterior of $\gamma$ and $G$ with the estimates $\hat{\delta}$ plugged in.
The next theorem states the posterior ratio consistency result of $G$ given $\gamma_0$ and $\hat{\delta}$, which implies the true graph $G_0$ is the mode of $\pi(G \mid \gamma_0, \hat{\delta}, Y, X)$ with probability tending to $1$.

\begin{theorem}[Posterior ratio consistency of $G$]\label{thm:G_cons}
	Suppose conditions (A2), (A4), (A5)  and $|\gamma_0| = O(1)$ hold.
	Then, for any $G \neq G_0$, 
	\bea
	\frac{\pi(\gamma_0 ,G \mid \hat{\delta},  Y, X) }{\pi(\gamma_0 ,G_0  \mid \hat{\delta}, Y, X)}
	&\overset{P}{\lra}& 0  \quad \text{ as } n \to \infty . 
	\eea
\end{theorem}

For any $\gamma$ and $G$, note that
\bea
\frac{\pi(\gamma, G \mid Y, X) }{\pi(\gamma_0, G \mid Y, X) }
&=& 
\frac{f(Y \mid X_\gamma, \gamma) \pi(X \mid G) \pi(\gamma\mid G) \pi(G) }{f(Y \mid X_{\gamma_0}, \gamma_0) \pi(X \mid G) \pi(\gamma_0\mid G) \pi(G) }   \\
&=& 
\frac{f(Y \mid X_\gamma, \gamma) \pi(X \mid \hat{\delta}, G) \pi(\gamma\mid G) \pi(G) }{f(Y \mid X_{\gamma_0}, \gamma_0) \pi(X \mid \hat{\delta}, G) \pi(\gamma_0\mid G) \pi(G) }   
\,\,=\,\,
\frac{\pi(\gamma, G \mid \hat{\delta}, Y, X) }{\pi(\gamma_0, G \mid \hat{\delta}, Y, X) } ,
\eea
where $f(Y \mid X_\gamma, \gamma) = \int f(Y \mid X_\gamma, \beta_\gamma) \pi(\beta_\gamma \mid \gamma) d \beta_\gamma$,  $\pi(X \mid G) = \int \pi(X \mid \Omega, G) \pi(\Omega \mid G) d\Omega$ and
 $\pi(X \mid \hat{\delta}, G) = \int \pi(X \mid \xi , \hat{\delta}, G ) \pi(\xi \mid G) d \xi$.
Then, by using the above equality, Theorems \ref{thm:gamma_cons} and \ref{thm:G_cons} imply joint posterior ratio consistency of $\gamma$ and $G$.
Corollary \ref{thm:joint_cons} states the joint selection consistency result.

\begin{corollary}[Joint posterior ratio consistency of $\gamma$ and $G$]\label{thm:joint_cons}
	Suppose conditions (A1)--(A6) hold.
	Then, $\gamma \neq \gamma_0$ and $G \neq G_0$,
	\bea
	\frac{\pi(\gamma ,G \mid \hat{\delta}, Y, X) }{\pi(\gamma_0 ,G_0 \mid \hat{\delta}, Y, X)}
	&\overset{P}{\lra}& 0  \quad \text{ as } n \to \infty . 
	\eea
\end{corollary}

In fact, the proposed method enjoys called joint selection consistency.
Theorem \ref{thm:joint_sel} shows that the joint posterior of $\gamma$ and $G$ given $\hat{\delta}$ is concentrated around the true values, $\gamma_0$ and $G_0$.
Joint selection consistency guarantees that the posterior mass assigned to $\gamma_0$ and $G_0$ converges to $1$ as $n\to\infty$.
This is a more powerful result than Corollary \ref{thm:joint_cons}, because joint selection consistency implies joint posterior ratio consistency, but not vice versa.

\begin{theorem}[Joint selection consistency of $\gamma$ and $G$]\label{thm:joint_sel}
	Suppose conditions (A1)--(A6) hold.
	Then,
	\bea
	\pi(\gamma_0 ,G_0 \mid \hat{\delta}, Y, X)  &\overset{P}{\lra}& 1 \quad \text{ as }  n\to\infty .
	\eea
\end{theorem}

\section{Simulation Studies}\label{sec:simulation}

In this section, we demonstrate the performance of the proposed method in various settings.
For $i = 1, 2, \ldots, n$, we simulate the data from
$Y_i =  I(Z_i \ge 0),$
where $Z_i = X_i\beta_{0} +\epsilon_i$, $\epsilon_i \sim N(0, 1)$ and $X_i = (x_{i1},  x_{i2}, \ldots, x_{ip})^T \overset{i.i.d.}{\sim} N_p(0, \Sigma_0)$, with the sample size $n$ and the number of predictors $p$.
Throughout the simulation study, we fix $n =100$.
If the atlas segments the brain into $p$ different anatomical sections, then, for example, we can consider $p$ as the number of brain regions.
In this case, the objective of joint inference would be to learn the abnormal functional activities among the significant brain regions that contribute to the disease onset.

Among these $p$ predictors, we assume that the first ten are active and consider the following four settings for the true coefficient vector $\beta_0$ to include different combinations of small and large signals.
\begin{itemize}
	\item Setting 1: All the nonzero entries of $\beta_{0}$ are set to 3.
	\item Setting 2: All the nonzero entries of $\beta_{0}$ are generated from $\mbox{Unif}(1.5, 3)$.
	\item Setting 3: All the nonzero entries of $\beta_{0}$ are set to 1.5.
	\item Setting 4: All the nonzero entries of $\beta_{0}$ are generated from $\mbox{Unif}(0.5,1.5)$.
\end{itemize}

For the true precision matrix $\Omega_0 = \sg_0^{-1}$, we consider the following four scenarios.
\begin{itemize}
	\item Scenario 1: 
	For $p=150$, we set all the diagonal entries to be 1 and $\Omega_{0, i1} = \Omega_{0, 1i} = 0.3$ for $i = 2, 3, \ldots, 10$, and set all the remaining entries to be 0. 
	
	\item Scenario 2: 
	For $p=150$, we consider a banded structures of $\Omega_0$ with all the unit diagonals, where $\Omega_{0, i, i+1} = \Omega_{0, i+1, i} = 0.3$, for $i = 1, 2, \ldots,p-1$.
	
	\item Scenario 3: 
	For $p=150$, we consider another banded structures of $\Omega_0$ with all the unit diagonals, where $\Omega_{0, i, i+1} = \Omega_{0, i+1, i} = 0.5, \Omega_{0, j, j+2} = \Omega_{0, j+2, j} = 0.25,$ for $i = 1, 2, \ldots,p-1, j = 1,2, \ldots, p-2$. 

	\item Scenario 4: 
	The true precision matrix $\Omega_0$ is set to be the same as in Scenario 1, but with $p=300$.
	This scenario will show the performance of the proposed method in high dimensions.
	
\end{itemize}


We will refer to our proposed joint selection method coupled with Bayesian spike and slab CONCORD as J.BSSC. 
In terms of variable selection, we first compare the performance of J.BSSC with other existing methods including Lasso \citep{tibshirani1996regression}, elastic net \citep{zou2005regularization} and the Bayesian joint selection method based on stochastic search structure learning (SSSL) \citep{Joint:Network:2016,wang2015scaling}, hereafter referred to as J.SSSL. 

The tuning parameters in Lasso and elastic net were chosen by 10-fold cross-validation.
For Bayesian methods, as discussed by \cite{Joint:Network:2016}, we suggest using the hyperparameters $a=2.75$ and $b=0.5$ for the MRF prior as default. 
Furthermore, to show the benefits of joint modeling, we also implement the setting with $b = 0$ for J.BSSC, which corresponds to the Bayesian method modeling the variable and precision matrix separately. 
The other hyperparameters were set at $a_0=0.1, b_0= 0.01,\tau^2= 1, q=0.005, r = 10^{-4}$ and  $s = 10^{-8}$. 
The initial state for $\gamma$ was set at $p$-dimensional zero vector, i.e., the empty model, while the initial state for the inverse covariance matrix was chosen by the graphical lasso (GLasso) \citep{Friedman2007glasso}. For posterior inference, $2,000$ posterior samples were drawn with a burn-in period of $2,000$.
As the final model, we chose the indices having posterior inclusion probability larger than $0.5$, which is called the median probability model. 
When the posterior probability of the posterior mode is larger than $0.5$, the median probability model corresponds to the posterior mode \citep{barbieri2004optimal}.

To evaluate the performance of  variable selection, the sensitivity, specificity,  Matthews correlation coefficient (MCC) and mean-squared prediction error (MSPE) are reported at Tables \ref{table:comp1} to \ref{table:comp4}.
The criteria are defined as
\bea
\text{Sensitivitiy}  &=&     \frac{TP}{TP+FN} ,   \\
\text{Specificity}  &=&    \frac{TN}{TN+FP}   ,  \\
\text{MCC}  &=&     \frac{TP \times TN - FP\times FN}{\sqrt{(TP+FP)(TP+FN)(TN+FP)(TN+FN)}}  ,	  \\ 
\text{MSPE}  &=&    \frac{1}{n_{\rm test}} \sum_{i=1}^{n_{\rm test}}     \big\{ \Phi(X_{{\rm test} , i } ^T \hat{\beta})  - Y_{{\rm test},i}  \big\}^2,  
\eea
where TP, TN, FP and FN are the true positive, true negative, false positive and false negative, respectively,
and $\hat{\beta}$ denotes the estimated coefficient based on each method.
%
%
For Bayesian methods, the usual GLM estimates based on the selected variables were used as $\hat{\beta}$.
We generated test samples and corresponding predictors $Y_{{\rm test}, 1}, Y_{{\rm test}, 2} \ldots, Y_{{\rm test}, n_{\rm test}}$ and $X_{{\rm test}, 1}, X_{{\rm test}, 2} \ldots, X_{{\rm test}, n_{\rm test}}$, respectively, with $n_{\rm test}=50$ to calculate the MSPE.

\begin{table}[!tb]
	\centering\footnotesize
	\caption{
		The summary statistics for Scenario 1 are represented for different settings, which corresponds to different choice of the true coefficient $\beta_0$.
	}\vspace{.15cm}
	\begin{tabular}{c c c c c | c c c c}
		\hline 
		& \multicolumn{4}{c}{ Setting 1 } & \multicolumn{4}{c}{ Setting 2 } \\ 
		& Sensitivity & Specificity & MCC & MSPE & Sensitivity & Specificity & MCC & MSPE \\ \hline
		J.BSSC $(b=\frac{1}{2})$  &0.87  &0.99  &0.87  &0.08 &0.79 &0.99 &0.80 &0.10 \\ 
		J.BSSC $(b=0)$ &0.61 &1 &0.76 &0.12  &0.46  &1   &0.63    &0.18   \\ 
		J.SSSL  &0.35 &0.97 &0.43 &0.24  &0.25  &0.97   &0.35    &0.29   \\ 
		Lasso  &0.72  &0.98  &0.69   &0.12  &0.80  &0.98   &0.75    &0.12   \\ 
		Elastic &0.90  &0.93  &0.62   &0.20  &1  &0.94   &0.70    &0.20   \\  \hline \hline 
		& \multicolumn{4}{c}{ Setting 3 } & \multicolumn{4}{c}{ Setting 4 } \\ 
		& Sensitivity & Specificity & MCC  & MSPE & Sensitivity & Specificity & MCC  & MSPE \\ \hline
		J.BSSC $(b=\frac{1}{2})$ &0.67  &0.99  &0.73   &0.14  &0.85  &0.99   &0.89    &0.08   \\ 
		J.BSSC $(b=0)$ &0.41  &0.99  &0.55   &0.18 &0.50  &1   &0.69    &0.16   \\ 
		J.SSSL  &0.20  &0.98  &0.32   &0.34  &0.32  &0.98   &0.40    &0.27   \\ 
		Lasso  &0.69  &0.98  &0.68   &0.14  &0.82  &0.98   &0.79    &0.09   \\ 
		Elastic &0.95  &0.94  &0.69   &0.20  &0.84  &0.97   &0.75    &0.19   \\  \hline  
	\end{tabular}\label{table:comp1}
\end{table}

\begin{table}[!tb]
	\centering\footnotesize
	\caption{
		The summary statistics for Scenario 2 are represented for different settings, which corresponds to different choice of the true coefficient $\beta_0$.
	}\vspace{.15cm}
	\begin{tabular}{c c c c  c | c c c c}
		\hline 
		& \multicolumn{4}{c}{ Setting 1 } & \multicolumn{4}{c}{ Setting 2 } \\ 
		& Sensitivity & Specificity & MCC & MSPE & Sensitivity & Specificity & MCC & MSPE \\ \hline
		J.BSSC $(b=\frac{1}{2})$ &1  &1  &1   &0.05  &1  &1   &1    &0.05   \\ 
		J.BSSC $(b=0)$ &0.90  &1  &0.95   &0.10  &0.88  &1   &0.90    &0.14   \\ 
		J.SSSL  &0.52  &0.98  &0.63   &0.20  &0.41  &0.97   &0.42    &0.21   \\ 
		Lasso  &0.74  &0.89  &0.44   &0.19  &0.66 &0.90 &0.41 &0.18   \\ 
		Elastic &0.70  &0.86   &0.40    &0.24  &0.50  &0.93   &0.38    &0.24   \\  \hline \hline 
		& \multicolumn{4}{c}{ Setting 3 } & \multicolumn{4}{c}{ Setting 4 } \\ 
		& Sensitivity & Specificity & MCC  & MSPE & Sensitivity & Specificity & MCC  & MSPE \\ \hline
		J.BSSC $(b=\frac{1}{2})$ &1  &0.99  &0.95   &0.08  &0.90  &0.99   &0.89    &0.11   \\ 
		J.BSSC $(b=0)$ &0.83  &0.99  &0.85   &0.12  &0.76   &0.99  &0.84   &0.18 \\ 
		J.SSSL  &0.36  &0.97  &0.39   &0.21   &0.34   &0.96 &0.29    &0.23   \\ 
		Lasso  &0.64  &0.92  &0.43   &0.19  &0.61  &0.89   &0.36    &0.17   \\ 
		Elastic &0.62  &0.89  &0.40   &0.24  &0.60  &0.86   &0.32    &0.23   \\  \hline  
	\end{tabular}\label{table:comp2}
\end{table}

\begin{table}[!tb]
	\centering\footnotesize
	\caption{
		The summary statistics for Scenario 3 are represented for different settings, which corresponds to different choice of the true coefficient $\beta_0$.
	}\vspace{.15cm}
	\begin{tabular}{c c c c  c | c c c c}
		\hline 
		& \multicolumn{4}{c}{ Setting 1 } & \multicolumn{4}{c}{ Setting 2 } \\ 
		& Sensitivity & Specificity & MCC & MSPE & Sensitivity & Specificity & MCC & MSPE \\ \hline
		J.BSSC $(b=\frac{1}{2})$ &0.92  &1  &0.96   &0.09  &0.77  &0.98   &0.71    &0.13   \\ 
		J.BSSC $(b=0)$ &0.90 &1 &0.94 &0.10  &0.52  &1  &0.65   &0.12   \\ 
		J.SSSL   &0.49  &0.98   &0.57  &0.20   &0.43  &0.98  &0.49   &0.20 \\ 
		Lasso  &0.51 &0.92 &0.35 &0.22  &0.41 &0.94 &0.33 &0.22   \\ 
		Elastic &0.55 &0.86 &0.36 &0.24  &0.48 &0.89 &0.32 &0.24   \\  \hline \hline 
		& \multicolumn{4}{c}{ Setting 3 } & \multicolumn{4}{c}{ Setting 4 } \\ 
		& Sensitivity & Specificity & MCC  & MSPE & Sensitivity & Specificity & MCC  & MSPE \\ \hline
		J.BSSC $(b=\frac{1}{2})$ &0.83  &0.96  &0.69   &0.18   &0.55  &0.99   &0.67    &0.17   \\ 
		J.BSSC $(b=0)$ &0.56 &1 &0.70 &0.17  &0.41  &1   &0.62    &0.18   \\ 
		J.SSSL &0.30  &0.97  &0.32   &0.24   &0.25   &0.97 &0.29    &0.26 \\
		Lasso  &0.43 &0.94 &0.33 &0.22 &0.56 &0.92 &0.38 &0.19   \\ 
		Elastic &0.52 &0.88 &0.37 &0.24 &0.55 &0.91 &0.37 &0.23   \\  \hline  		
	\end{tabular}\label{table:comp3}
\end{table}

\begin{table}[!tb]
	\centering\footnotesize
	\caption{
		The summary statistics for Scenario 4 are represented for different settings, which corresponds to different choice of the true coefficient $\beta_0$.
	}\vspace{.15cm}
	\begin{tabular}{c c c c  c | c c c c}
		\hline 
		& \multicolumn{4}{c}{ Setting 1 } & \multicolumn{4}{c}{ Setting 2 } \\ 
		& Sensitivity & Specificity & MCC & MSPE & Sensitivity & Specificity & MCC & MSPE \\ \hline
		J.BSSC $(b=\frac{1}{2})$ &0.78  &1  &0.86   &0.07  &0.74  &1   &0.85    &0.08   \\ 
		J.BSSC $(b=0)$ &0.57 &1 &0.72   &0.07  &0.48  &1   &0.65    &0.13   \\ 
		J.SSSL  &0.40  &0.99  &0.43   &0.16  &0.31  &0.99  &0.33   &0.19    \\ 
		Lasso  &0.70  &0.99  &0.73   &0.08  &0.73  &0.98   &0.61    &0.06   \\ 
		Elastic &0.79  &0.99  &0.75   &0.20  &0.75  &0.99   &0.77    &0.18   \\  \hline \hline 
		& \multicolumn{4}{c}{ Setting 3 } & \multicolumn{4}{c}{ Setting 4 } \\ 
		& Sensitivity & Specificity & MCC  & MSPE & Sensitivity & Specificity & MCC  & MSPE \\ \hline
		J.BSSC $(b=\frac{1}{2})$ &0.70  &1  &0.78 &0.12  &0.64  &1   &0.72    &0.17   \\ 
		J.BSSC $(b=0)$ &0.45  &1  &0.61   &0.15  &0.38  &1   &0.51    &0.16   \\ 
		J.SSSL    &0.25  &0.98   &0.29    &0.19 &0.22  &0.98  &0.24   &0.21 \\ 
		Lasso  &0.72 &0.98  &0.61   &0.11  &0.70  &0.97   &0.60    &0.13   \\ 
		Elastic &0.67  &0.98  &0.59   &0.18  &0.59  &0.99   &0.65    &0.19   \\  \hline  
	\end{tabular}\label{table:comp4}
\end{table}

The sensitivity, specificity, MCC and MSPE, under different scenarios, are reported at Tables \ref{table:comp1}--\ref{table:comp4} to evaluate the variable selection performance.
We notice that compared to regularization methods (Lasso and elastic net), the proposed joint selection approach (J.BSSC) tends to have better specificity and MCC. 
The poor specificity of the regularization methods has also been discussed in previous literature in the sense that selection of the regularization parameter using cross-validation is optimal with respect to prediction but tends to include too many noise predictors \citep{meinshausen2006high}.
This leads to relatively larger numbers of errors for the regularization methods compared with those for the Bayesian joint selection methods. 
Among all Bayesian approaches, under most of settings, the proposed J.BSSC approach (with $b = 0.5$ or $b = 0$) outperforms J.SSSL based on all criteria, which shows the benefit of the proposed joint method incorporating the graph structure through the CONCORD generalized likelihood. 
Interestingly, compared with J.SSSL that adopts the Metropolis-Hastings algorithm for variable selection, the performance of the proposed Gibbs sampler is significantly better in terms of almost all the measures.
Furthermore, J.BSSC with $b = 0.5$ tends to have a slightly lower specificity but significantly higher sensitivity, MCC and lower MSPE compared with J.BSSC with $b = 0$.
This could be caused by the proposed method frequently visiting graph-linked variables due to the MRF prior. 
We also found that the proposed J.BSSC overall works better than other methods especially in the strong signal setting (i.e., Setting 1). 
This is because as signal strength gets stronger, the consistency conditions of our method are easier to satisfy which leads to better performance. 
To sum up, the above observation indicates that the proposed method can achieve good variable selection performance under a variety of configurations with different data generation mechanisms.


\begin{table}[!tb]
	\centering
	\caption{
		The summary statistics for graph selection under Setting 1 and Scenario 1 are represented.
	}\vspace{.15cm}
	\begin{tabular}{ccccc}
		\hline
		& Sensitivity & Specificity & MCC  & \#Error \\ \hline
		J.BSSC & 1           & 1           & 0.90 & 2       \\
		J.SSSL &1             &1             &0.87      &3         \\
		GLasso & 1           & 0.98        & 0.19 & 239     \\
		CLIME  & 1           & 0.98        & 0.18 & 256     \\
		TIGER  & 1           & 1           & 0.73 & 8      \\ \hline
	\end{tabular} \label{table:comp5}
\end{table}

\begin{table}[!tb]
	\centering
	\caption{
		The summary statistics for precision matrix estimation under Setting 1 and Scenario 1 are represented.
	}\vspace{.15cm}
	\begin{tabular}{ccccc}
		\hline
		& $E_1$   & $E_2$   & $E_3$   & $E_4$   \\ \hline
		J.BSSC               & 0.13 & 0.21 & 0.08 & 0.28 \\
		J.SSSL               &8.01      &7.26      &1.86      &11.95      \\
		GLasso               & 0.37 & 0.24 & 0.19 & 0.19 \\
		CLIME                & 1.51 & 2.22 & 0.58 & 4.16 \\
		TIGER                & 1.47 & 1.91 & 0.31 & 3.48 \\ \hline
	\end{tabular} \label{table:comp6}
\end{table}

We also briefly present the performance of graph selection and precision matrix estimation for J.BSSC.
We compare the performance of J.BSSC with other existing methods including J.SSSL \citep{Joint:Network:2016,wang2015scaling}, GLasso \citep{Friedman2007glasso}, the constrained $\ell_1$-minimization for inverse matrix estimation (CLIME) \citep{cai2011constrained} and the tuning-insensitive approach for optimally estimating Gaussian graphical models (TIGER) \citep{liu2017tiger}. 
The tuning parameters for GLasso and TIGER were chosen by the criterion of stability approach to regularization selection (StARS) \citep{Liu2010STARS}. 
We used 10-fold cross-validation to select the penalty parameter for CLIME.
For GLasso and TIGER, the final models were constructed by collecting the nonzero entries in the estimated precision matrix. 
In our simulation settings, CLIME could not produce exact zeros, so we chose the final graph estimate by thresholding the absolute values of the estimated precision matrix at 0.1. 

To evaluate the performance of graph selection and precision matrix estimation, we report the results at Tables \ref{table:comp5} and \ref{table:comp6}, where each simulation setting is repeated for 20 times. 
The results under different scenarios are omitted because they gave similar conclusions, and only the results under Scenario 1 are presented in the tables.
In Table \ref{table:comp5}, \#Error denotes the number of errors, i.e., FP+FN.
For a matrix norm $\|\cdot\|$ and an estimator $\hat{\Omega}$, the relative error $\|\Omega_0 - \hat{\Omega}\| / \|\Omega_0\|$ is chosen as a criterion.
In Table \ref{table:comp6}, $E_1$, $E_2$, $E_3$ and $E_4$ represent the relative errors based on the matrix $\ell_1$-norm, the matrix $\ell_2$-norm (spectral norm), the vector $\ell_2$-norm (Frobenius norm) and the vector $\ell_\infty$-norm (entrywise maximum norm), respectively.

Based on the results in Table \ref{table:comp5}, in terms of graph selection, joint selection approaches (J.BSSC and J.SSSL) outperform other contenders estimating a graph $G$ without incorporating information about $\gamma$.
This suggests that joint selection using an MRF prior can benefit not only variable selection performance but also graph selection performance.
Furthermore, Table \ref{table:comp6} shows that J.BSSC performs significantly better than J.SSSL in terms of precision matrix estimation.
In fact, J.BSSC also outperforms the other contenders for all the criteria considered.
Therefore, it can be interpreted that joint selection improves the estimation performance, and in particular, it is more preferable to use CONCONRD for precision matrix estimation.

\begin{figure}[tb]
	\centering
	\includegraphics[width=100mm]{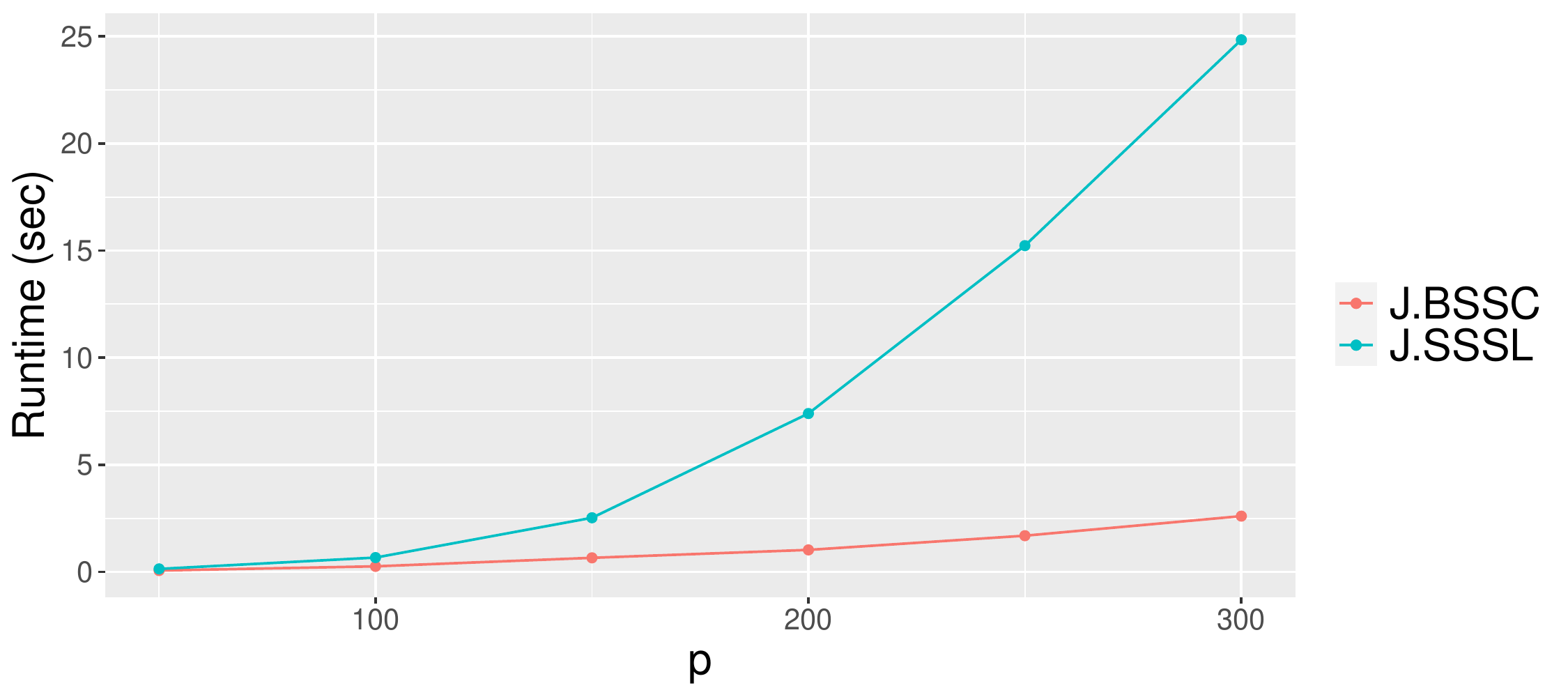}
	\caption{The comparison of average wall-clock seconds per iteration under different dimensions.}
	\label{run_time}
\end{figure}

In addition, as noted in \cite{jalali2020bconcord}, BSSC is computationally much more efficient compared with SSSL.
In Figure \ref{run_time}, we plot the run time comparison between J.BSSC and J.SSSL under different values of $p$ coded in R. 
The averaged computation times for J.BSSC are significantly smaller than those for J.SSSL, and the gap between the two gets larger as $p$ grows.
Even in terms of the memory requirement, J.BSSC needs a significantly smaller memory than SSSL. 
For example, J.SSSL requires more than 20 GB  while J.BSSC achieves the goal with 0.22 GB of memory when  $p = 300$. 
Furthermore, based on asymptotic results, one can expect that our method will give accurate inference results as we have more observations, while asymptotic properties of the Bayesian method proposed by \cite{Joint:Network:2016}  are still in question.

\section{Aberrant Functional Activities in the Parkinson's Disease Cohort} \label{sec:PD}

Parkinson’s disease (PD) was first described by Dr. James Parkinson in 1817 as ``shaking palsy". It is a chronic, progressive neurodegenerative disease characterized by both motor and nonmotor features. As one of the most common neurodegenerative disorders, the disease has a significant clinical impact on patients, families, and caregivers through its progressive degenerative effects on mobility and muscle control. Research suggests that the pathophysiological changes associated with PD may start before the onset of motor features and may include a number of nonmotor presentations, such as sleep disorders, depression, and cognitive changes. Evidence for this preclinical phase has driven the enthusiasm for research that focuses on early diagnosis and preventive therapies of PD \citep{Schrag2015}.

In recent years, neuroimaging has been increasingly employed to aid the risk stratification in PD.  Among a variety of neuroimaging technologies, resting-state fMRI (rs-fMRI) is regarded as a promising technique for precisely locating the abnormal spontaneous activities in neuropsychological disease \citep{Wang2019}. Several rs-fMRI-based methods including regional homogeneity (ReHo), the amplitude of low-frequency fluctuation, and functional connectivity provide a task-free approach to explore spontaneous brain activity and connectivity among networks in different brain regions of PD patients. In this section, we apply the proposed joint selection method to rs-fMRI data for simultaneously identifying aberrant functional brain activities and inferring the underlying functional brain network to aid the diagnosis of PD \citep{wei2017aberrant,Cao2020}. 

\subsection{Subjects and data preprocessing}
This study was approved by the Medical Research Ethical Committee of Nanjing Brain Hospital (Nanjing, China) in accordance with the Declaration of Helsinki, and written informed consent was obtained from all subjects. Seventy PD patients and fifty healthy controls (HCs) were recruited. Image data were acquired using a Siemens 3.0-Tesla signal scanner (Siemens, Verio, Germany) in the department of radiology within Nanjing Brain Hospital. Functional imaging data were collected transversely by using a gradient-recalled echo-planar imaging pulse sequence and retrieved from the archive by neuroradiologists. Image preprocessing steps including slice-timing correction and spatial normalization were carried out using the Data Processing Assistant for Resting-State fMRI based on Statistical Parametric Mapping (SPM12) operated on the Matlab platform \citep{DPARSF}.

\subsection{Image feature extraction}
\cite{Zang:Jiang} proposed the method of Regional Homogeneity (ReHo) to analyze characteristics of regional brain activity and to reflect the temporal homogeneity of neural activity. ReHo is defined as a voxel-based measure of brain activity which evaluates the similarity or synchronization between the time series of a given voxel and its nearest neighbors. Abnormal ReHo signals, which are associated with changes in neuronal activity in local brain regions, may be exploited  to analyze the abnormal brain activities and to depict the dynamic brain functional connectivities \citep{ReHo2019,Deng2016}. In particular, we focus on the mReHo maps obtained by dividing the mean ReHo of the whole brain within each voxel in the ReHo map. We further segmented the mReHo maps and extracted all the 112 ROI signals based on the Harvard-Oxford atlas (HOA) using the Resting-State fMRI Data Analysis Toolkit \citep{REST}.

\subsection{Model fitting}
We now consider a probit regression model with the binary disease indicator as an outcome and 112 ReHo radiomic variables as predictors. Various models including the proposed method and other competing approaches will then be implemented to classify subjects based on these extracted features and to learn functional connectivities of the brain. The dataset is randomly divided into a training set (80\%) and a testing set (20\%) while maintaining the PD:HC ratio in both sets. The hyperparameters for all methods are set as in simulation studies. For Bayesian methods, we first obtain the identified variables and then evaluate the testing set performance using standard GLM estimates based on the selected features. The penalty parameters in all frequentist methods are tuned via 10-fold cross validation in the training set. The final prediction results based on the testing set for both Bayesian and frequentist approaches are evaluated using a common threshold 0.5.

\subsection{Results}
In terms of discriminative radiomic features, our method is able to identify abnormal functional brain activities for PD that occur in the regions of interest including right superior frontal gyrus (F1.R), left middle temporal gyrus, anterior division (T2a.L), left angular gyrus (AG.L), right angular gyrus (AG.R), right temporal fusiform cortex, anterior division (TFa.R), right occipital fusiform gyrus (OF.R), left frontal operculum cortex (FO.L) and left putamen (Put.L). 
In Figure \ref{brain_network}, we plot the inferred functional brain network overlaid with selected nodes that correspond to the aforementioned brain regions. The predictive performance of various methods in the test set is summarized in Table \ref{table:MRI}. We can tell from Table \ref{table:MRI} that the predictive performance of the proposed joint selection approach based on BSSC is overall better than that of all the other methods. The proposed J.BSSC approach has higher sensitivity and lower MSPE compared with all the other methods, but yields a lower specificity than Lasso. Based on the most comprehensive measure MCC, our method outperforms all the other methods.

\begin{table}[H]
	\centering\footnotesize
	\caption{
		The summary statistics for prediction performance on the testing set for all methods.}
	\vspace{.15cm}
	\begin{tabular}{c c c c c}
		\hline 
		& Sensitivity & Specificity &MCC & MSPE  \\ \hline
		J.BSSC $(b=\frac{1}{2})$ &0.92 &0.82 &0.74 &0.09   \\ 
		J.BSSC $(b=0)$ &0.67  &0.73  &0.39  &0.18 \\ 
		J.SSSL   &0.58  &0.73  &0.31  &0.24 \\ 
		Lasso  &0.67  &0.91  &0.59  &0.16 \\ 
		Elastic &0.75  &0.82  &0.57  &0.16  \\\hline 
	\end{tabular}\label{table:MRI}
\end{table}

\begin{figure} [tb]
	\centering
	\includegraphics[width=100mm]{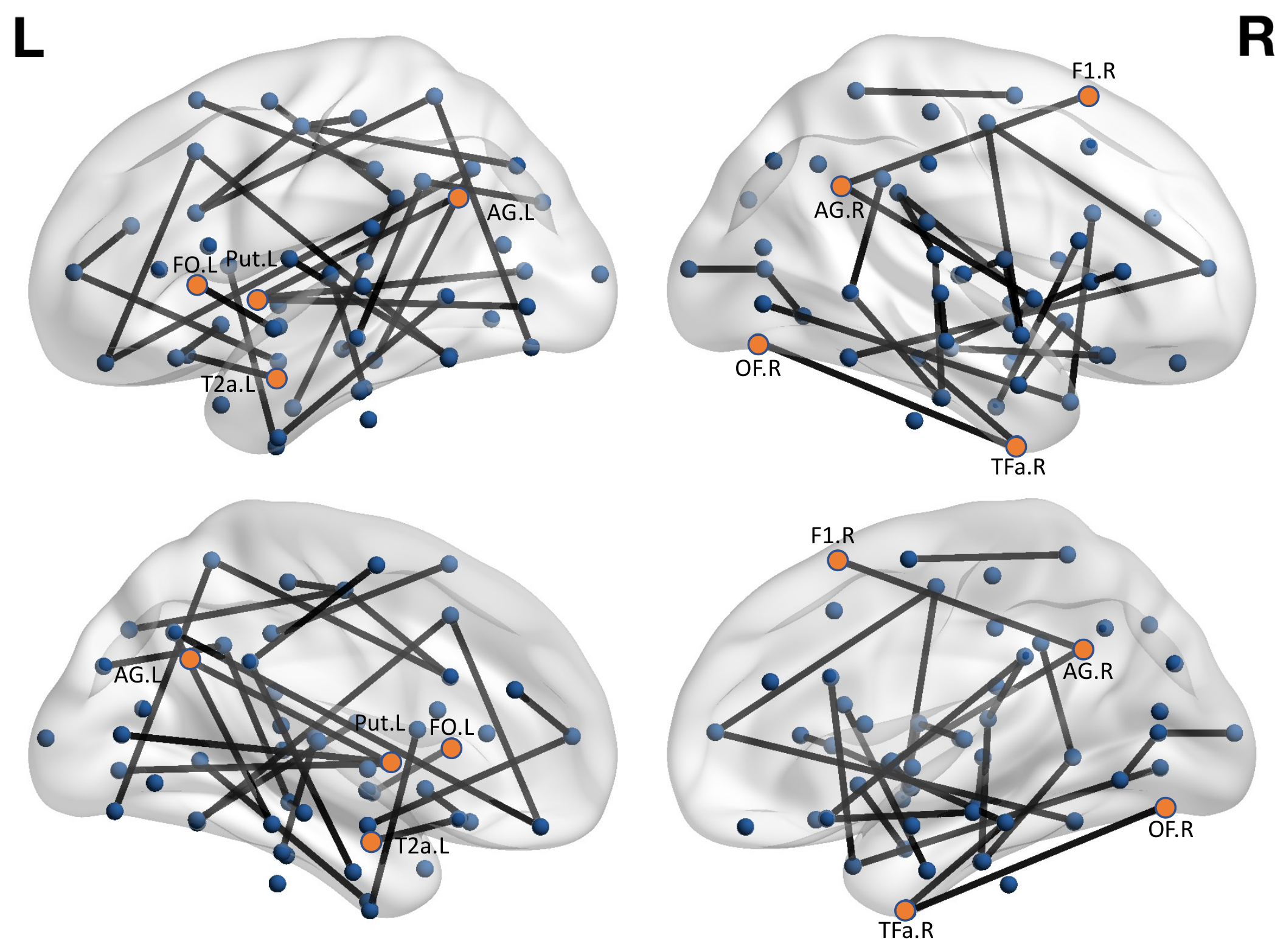}
	\caption{The lateral and medial view of the functional brain network inferred by J.BSSC. Nodes selected by J.BSSC are marked in orange.}
	\label{brain_network}
\end{figure}

Furthermore, J.BSSC identifies regions of interest that are coherent with the altered functional features in cortical and subcortical regions discovered in previous studies \citep{fmri:1,fmri:2,fmri:3}. 
These findings suggest disease-related alterations of functional activities that provide physicians sufficient information to get involved with early diagnosis and treatment. The inferred functional brain connectivities also seem plausible and are primarily located in the typical resting-state network (RSN) including default-mode network (DMN), visual network (VIN) and basal ganglia network (BGN). The identified regions in DMN include the left middle temporal gyrus, anterior division and angular gyrus. We also discover abnormal VIN in the right temporal fusiform cortex, anterior division and right occipital fusiform gyrus, as well as unusual BGN in the left putamen. RSN reflects the spontaneous neural activities of the blood oxygenation level-dependent signals between temporally correlated brain regions. Compared with the control group, the DMN plays a crucial role in neurodegenerative disorders and normal aging. Several fMRI studies have indicated that the DMN was injured before the cognitive decline in PD \citep{Sandronee172,Koshimori2016}. The BGN has also been observed in pathologies with motor control and altered neurotransmitter systems of dopaminergic processes \citep{Griffanti2018,DeMicco2019}.
A previous study on functional connectivity markers in advanced PD also found functional connectivity features located in the VIN and cerebellar networks that are significantly relevant to classification and provide preliminary evidence that can characterize PD patients compared with HCs \citep{Lin2020102130}. In conclusion, the radiomics-based joint selection approach proposed in this paper has shown that high-order radiomic features that quantify functional brain connectivities and activities can be used for the diagnosis of PD with satisfactory prediction accuracy.

\section{Discussion} \label{sec:discussion}

We propose a Bayesian joint selection method for probit models.
Although it should be rigorously investigated, it is possible to extend the proposed method to other GLMs with network-structured predictors and binary responses.
For example, an extension to logistic regression models, in terms of computation, is straightforward by approximating a logistic distribution to mixture of normal distributions \citep{Albert:Chib:1993,o2004bayesian}.
This approximation enables us to derive a similar Gibbs sampler presented in Section \ref{sec:computation} with some minor changes; for example, see \cite{lee2021bayesian}.
Furthermore, in theoretical aspect, it is highly expected that joint selection consistency (Theorem \ref{thm:joint_sel}) can be achieved in logistic regression models with CONCORD generalized likelihood by applying the techniques in \cite{lee2021bayesian}, which efficiently control the score function and Hessian matrix of logistic models.

Theoretical results in this paper, except Theorem \ref{thm:gamma_cons}, are based on the conditional posteriors given accurate estimates of diagonal entries, $\hat{\delta}$.
This is because we adopt the selection consistency result in \cite{jalali2020bconcord}.
It would be interesting to investigate whether one can obtain selection consistency without conditioning $\hat{\delta}$ to conduct a fully Bayesian inference.
This would need a significant amount of technical modification, so we leave it as future work.

Furthermore, by using CONCORD generalized likelihood, we can enjoy fast computational speed but at the cost of possibly losing the positive definiteness of the precision matrix.
Although it does not harm the primary goal of this paper, the selection of the support of the precision matrix and coefficient vector, it will obviously not be satisfactory when the estimation of the precision matrix is of interest.
Thus, modifying the CONCORD algorithm to ensure positive definiteness of the precision matrix while maintaining fast computation would be another possible direction of future work.


\appendix

\section{Proofs}\label{sec:proofs}

\noindent{\bf Notation.} 
In the rest of the paper, we denote $\bfY_n \equiv Y = (Y_1, Y_2, \ldots, Y_n)^T \in \bbR^n$ and $\bfX_n \equiv X = (X_1, X_2, \ldots, X_n)^T \in \bbR^{n\times p}$.

\noindent{\bf Score function and Hessian matrix.}
For any $\gamma$, let $\eta_{i,\gamma} = X_{i,\gamma}^T \beta_\gamma$.
Then,  the log-likelihood function is 
\bea
L_n(\beta_\gamma)  &=& \sum_{i=1}^n \Big[  Y_i \log \Phi ( X_{i,\gamma}^T \beta_\gamma  )  + (1-Y_i)  \log \big\{  1- \Phi( X_{i,\gamma}^T \beta_\gamma )  \big\}   \Big] .
\eea
The score function and Hessian matrix are given by 
\bea
s_n(\beta_\gamma)  &=& \frac{\partial}{\partial \beta_\gamma} L_n(\beta_\gamma) 
\,\,=\,\, \sum_{i=1}^n X_{i,\gamma} \Big\{  Y_i \frac{\phi( \eta_{i,\gamma} ) }{\Phi(\eta_{i,\gamma}) }  -  (1- Y_i) \frac{\phi(\eta_{i,\gamma}) }{1- \Phi(\eta_{i\gamma})}   \Big\} ,\\
&\equiv&  \bfX_\gamma ^T D(\beta_\gamma) \sg(\beta_\gamma)^{-1/2} \{Y- \mu(\beta_\gamma)\}   , \\
H_n(\beta_\gamma)  &=& \frac{\partial^2}{\partial \beta_\gamma \partial \beta_\gamma^T } L_n(\beta_\gamma) \\
&=&  \sum_{i=1}^n X_{i,\gamma} X_{i,\gamma}^T \bigg[  Y_i \Big\{  \frac{\eta_{i,\gamma} \phi(\eta_{i,\gamma})}{\Phi(\eta_{i,\gamma}) } +    \frac{\phi(\eta_{i,\gamma})^2}{\Phi(\eta_{i,\gamma})^2 }    \Big\}    + (1-Y_i) \Big\{  \frac{-\eta_{i,\gamma} \phi(\eta_{i,\gamma})}{1-\Phi(\eta_{i,\gamma}) } +    \frac{\phi(\eta_{i,\gamma})^2}{ (1- \Phi(\eta_{i,\gamma}))^2 }  \Big\}  \bigg]  \\
&\equiv& \sum_{i=1}^n X_{i,\gamma} X_{i,\gamma}^T \cdot \psi_i(\beta_\gamma)  \,\,\equiv\,\, \bfX_\gamma^T  \Psi_\gamma  \bfX_\gamma
\eea
where $D(\beta_\gamma)  = diag(d_{i}(\beta_\gamma) ) \in \bbR^{n\times n}$ with $d_{i}(\beta_\gamma)  = \phi(\eta_{i,\gamma}) /\sqrt{\Phi(\eta_{i,\gamma}) (1- \Phi(\eta_{i,\gamma}) ) }$, and $\mu(\beta_\gamma) = (\mu_i(\beta_\gamma)  ) \in \bbR^n$ with $\mu_i(\beta_\gamma)  = \Phi(\eta_{i,\gamma} )$, and $\sg(\beta_\gamma) = diag( \sigma_i^2(\beta_\gamma)   ) \in \bbR^{n\times n}$ with $\sigma_i^2 (\beta_\gamma)  =  \Phi(\eta_{i,\gamma} ) (1-  \Phi(\eta_{i,\gamma} ))$, and $\Psi_\gamma  = diag( \psi_i(\beta_\gamma) ) \in \bbR^{n\times n}$.
For simplicity, let $\mu = ( \mu_i( \beta_{0}) )$ and $\sg = diag( \sigma_i^2(\beta_0) )$.

\begin{proof}[Proof of Theorem \ref{thm:gamma_cons}]
	Note that for any $G$,
	\bea
	\pi( \gamma , G \mid \bfY_n , \bfX_n )
	&\propto& f( \bfY_n \mid   \bfX_\gamma , \gamma ) \pi( \bfX_n \mid   G ) \pi(\gamma \mid G)  \pi(G) ,
	\eea
	where 
	\bea
	f( \bfY_n \mid \bfX_\gamma , \gamma ) &=& \int f( \bfY_n \mid \bfX_\gamma , \beta_\gamma) \pi(\beta_\gamma \mid \gamma ) d\beta_\gamma  \\
	&\equiv& \int \exp \big\{ L_n (\beta_\gamma) \big\} \, (2\pi \tau^2)^{- |\gamma|/2 }   \exp \Big( - \frac{1}{2\tau^2} \|\beta_\gamma \|_2^2 \Big) d \beta_\gamma .
	\eea
	Thus, 
	\bea
	\frac{\pi( \gamma , G \mid \bfY_n , \bfX_n )}{\pi( \gamma_0  , G \mid \bfY_n , \bfX_n )}
	&=& \frac{f( \bfY_n \mid \bfX_\gamma , \gamma ) \pi(\gamma \mid G) }{f( \bfY_n \mid \bfX_{\gamma_0} , \gamma_0 ) \pi(\gamma_0 \mid G)}  
	\eea
	and $\pi( \bfX_n \mid  G ) = \int \pi(\bfX_n \mid \Omega,   G) \pi(\Omega \mid G)  d \Omega$.

	First, we focus on overfitted models, $M_1 = \{ \gamma : \gamma \supsetneq \gamma_0 , \, |\gamma| \le R_2 \}$.
	By Taylor's expansion of $L_n(\beta_\gamma)$ around the MLE of $\beta_\gamma$ under the model $\gamma$, say $\hat{\beta}_\gamma$,
	\bea
	L_n(\beta_\gamma) - L_n(\hat{\beta}_\gamma) 
	&=& - \frac{1}{2} (\beta_\gamma - \hat{\beta}_\gamma)^T H_n(\tilde{\beta}_\gamma) (\beta_\gamma - \hat{\beta}_\gamma)
	\eea
	for some $\tilde{\beta}_\gamma$ such that $\|\tilde{\beta}_\gamma - \hat{\beta}_\gamma\|_2 \le \| \beta_\gamma - \hat{\beta}_\gamma \|_2$.
	For any $\beta_\gamma$ such that $\|\beta_\gamma - \beta_{0,\gamma} \|_2 \le C \sqrt{|\gamma| \log p /n } \equiv C w_n$ for some constant $C>0$, by Lemma \ref{lem:betahat}, 
	\bea
	\| \tilde{\beta}_\gamma - \beta_{0,\gamma} \|_2 
	&\le& \|\tilde{\beta}_\gamma - \hat{\beta}_\gamma \|_2 + \| \hat{\beta}_\gamma - \beta_{0,\gamma} \|_2 \\
	&\le& \|{\beta}_\gamma - \hat{\beta}_\gamma \|_2 + \| \hat{\beta}_\gamma - \beta_{0,\gamma} \|_2 \\
	&\le& \|{\beta}_\gamma - {\beta}_{0,\gamma} \|_2 + 2 \| \hat{\beta}_\gamma - \beta_{0,\gamma} \|_2
	\,\,\le\,\, 3 C w_n
	\eea
	uniformly for all $\gamma \in M_1$ with probability at least $1- 2 \exp (-cn)$ for some constant $c>0$.
	Thus, by Lemma \ref{lem:hessian}, 
	\bea
	L_n(\beta_\gamma) - L_n(\hat{\beta}_\gamma) 
	&\le& 
	-\frac{1-\epsilon}{2} (\beta_\gamma - \hat{\beta}_\gamma)^T H_n({\beta}_{0,\gamma}) (\beta_\gamma - \hat{\beta}_\gamma)
	\eea
	for some small constant $\epsilon>0$.
	For any $\beta_\gamma$ such that $\|\beta_\gamma - \hat{\beta}_\gamma \|_2 = C w_n /2$, we have
	\bea
	L_n(\beta_\gamma) - L_n(\hat{\beta}_\gamma) 
	&\le& - \frac{1-\epsilon}{2} \|\beta_\gamma - \hat{\beta}_\gamma \|_2^2 \, \lambda_{\min} \big( H_n(\beta_{0,\gamma} ) \big)  \\
	&\le& - \frac{1-\epsilon}{8} C^2 \lambda |\gamma| \log p \,\,\lra\,\, -\infty \,\, \text{ as } n\to \infty   
	\eea
	with probability at least $1- 2 \exp(-cn)$, by Lemma \ref{lem:hessian_lower}.
	Note that it also holds for any $\beta_\gamma$ such that $\| \beta_\gamma - \hat{\beta}_\gamma \|_2 > C w_n /2$ due to the concavity of $L_n(\cdot)$ and the fact that $\hat{\beta}_\gamma$ maximizes $L_n(\beta_\gamma)$.

	Define $B_\gamma = \{ \beta_\gamma : \|\beta_\gamma - \hat{\beta}_\gamma \|_2 \le C w_n /2 \}$, then $B_\gamma \subset \{ \beta_\gamma : \|\beta_\gamma - \beta_{0,\gamma} \|_2 < C w_n \}$ with probability at least $1- 2\exp(-cn)$ uniformly in $\gamma \in M_1$.
	Therefore, for any $\gamma \in M_1$, with probability at least $1- 2 \exp(-cn)$, 
	\bea
	&&  f( \bfY_n \mid \bfX_\gamma , \gamma ) \pi(\gamma \mid G) \\
	&=&    \int \exp \big\{ L_n (\beta_\gamma) \big\} \, (2\pi \tau^2)^{- |\gamma|/2 }   \exp \Big( - \frac{1}{2\tau^2} \|\beta_\gamma \|_2^2 \Big) d \beta_\gamma  \pi(\gamma \mid G)  \\
	&\le& (2\pi \tau^2)^{- |\gamma|/2} \pi(\gamma \mid G) \exp \big\{  L_n(\hat{\beta}_\gamma) \big\} \bigg[  \exp \Big( - \frac{1-\epsilon}{8} C^2 \lambda |\gamma| \log p \Big)  \, \int_{B_\gamma^c} \exp \Big( - \frac{\|\beta_\gamma\|_2^2}{2\tau^2 } \Big) d \beta_\gamma  \\
	&& \quad\quad\quad + \,\, \int_{B_\gamma} \exp \Big\{  - \frac{1-\epsilon}{2} (\beta_\gamma - \hat{\beta}_\gamma)^T H_n( \beta_{0,\gamma} ) (\beta_\gamma - \hat{\beta}_\gamma)  - \frac{\|\beta_\gamma\|_2^2}{2\tau^2}   \Big\}  d \beta_\gamma \bigg] ,
	\eea
	where 
	\bea
	&& \int_{B_\gamma} \exp \Big\{  - \frac{1-\epsilon}{2} (\beta_\gamma - \hat{\beta}_\gamma)^T H_n( \beta_{0,\gamma} ) (\beta_\gamma - \hat{\beta}_\gamma)  - \frac{\|\beta_\gamma\|_2^2}{2\tau^2}   \Big\}  d \beta_\gamma \\
	&\le& (2\pi)^{|\gamma|/2} \det \Big\{ (1-\epsilon) H_n(\beta_{0,\gamma}) + \tau^{-2} I_{|\gamma|}  \Big\}^{-1/2}
	\eea
	and
	\bea
	&& \exp \Big( - \frac{1-\epsilon}{8} C^2 \lambda |\gamma| \log p \Big)  \, \int_{B_\gamma^c} \exp \Big( - \frac{\|\beta_\gamma\|_2^2}{2\tau^2 } \Big) d \beta_\gamma \\
	&\le& \exp \Big( - \frac{1-\epsilon}{8} C^2 \lambda |\gamma| \log p + \frac{|\gamma|}{2} \log \tau^2  \Big) (2\pi)^{|\gamma|/2}   \\
	&\le& \exp \big( - C' |\gamma| \log p  \big) (2\pi)^{|\gamma|/2} \det \Big\{ (1-\epsilon) H_n(\beta_{0,\gamma}) + \tau^{-2} I_{|\gamma|}  \Big\}^{1/2 -1/2} \\
	&\le& \exp \big( - C'' |\gamma| \log p     \big) (2\pi)^{|\gamma|/2}   \det \Big\{ (1-\epsilon) H_n(\beta_{0,\gamma}) + \tau^{-2} I_{|\gamma|}  \Big\}^{ -1/2} 
	\eea
	for some positive constants $C'$ and $C''$.
	Hence, we have 
	\bean\label{gamma_M1}
	&&  f( \bfY_n \mid \bfX_\gamma , \gamma ) \pi(\gamma \mid G) \nonumber \\
	&\le& (\tau^2)^{-|\gamma|/2} \pi(\gamma \mid G) \exp \big\{ L_n(\hat{\beta}_\gamma) \big\} \,  \det \Big\{ (1-\epsilon) H_n(\beta_{0,\gamma}) + \tau^{-2} I_{|\gamma|}  \Big\}^{ -1/2}  \big( 1 + o(1) \big)
	\eean
	for any $\gamma \in M_1$, with probability at least $1- 2 \exp(-cn)$.
	
	On the other hand, 
	\bea
	&& f( \bfY_n \mid \bfX_{\gamma_0} , \gamma_0 ) \pi(\gamma_0 \mid G) \\
	&=& 
	 \int \exp \big\{ L_n (\beta_{\gamma_0}) \big\} \, (2\pi \tau^2)^{- |\gamma_0|/2 }   \exp \Big( - \frac{1}{2\tau^2} \|\beta_{\gamma_0} \|_2^2 \Big) d \beta_{\gamma_0}  \pi(\gamma_0 \mid G)  \\
	 &\ge& (2\pi \tau^2)^{-|\gamma_0|/2} \pi(\gamma_0\mid G) \exp \big\{ L_n( \hat{\beta}_{\gamma_0} ) \big\}  \int_{B_{\gamma_0}} \exp \Big\{  - \frac{1+\epsilon}{2} (\beta_{\gamma_0} - \hat{\beta}_{\gamma_0})^T H_n( \beta_{0,{\gamma_0}} ) (\beta_{\gamma_0} - \hat{\beta}_{\gamma_0})  - \frac{\|\beta_{\gamma_0}\|_2^2}{2\tau^2}   \Big\}  d \beta_{\gamma_0}  ,
	\eea
	where $A = (1+\epsilon) H_n( \beta_{0, \gamma_0} )$ and
	\bea
	&& \int_{B_{\gamma_0}} \exp \Big\{  - \frac{1+\epsilon}{2} (\beta_{\gamma_0} - \hat{\beta}_{\gamma_0})^T H_n( \beta_{0,{\gamma_0}} ) (\beta_{\gamma_0} - \hat{\beta}_{\gamma_0})  - \frac{\|\beta_{\gamma_0}\|_2^2}{2\tau^2}   \Big\}  d \beta_{\gamma_0}  \\
	&=& (2\pi)^{-|\gamma_0|/2} \det \big( A + \tau^{-2} I_{|\gamma_0|} \big)^{-1/2} \exp \Big[  - \frac{1}{2} \hat{\beta}_{\gamma_0}^T \Big\{  A -  A ( A + \tau^{-2} I_{|\gamma_0|} )^{-1} A   \Big\}  \hat{\beta}_{\gamma_0}   \Big] \\
	&\gtrsim& (2\pi)^{-|\gamma_0|/2} \det \big( A + \tau^{-2} I_{|\gamma_0|} \big)^{-1/2} 
	\eea
	uniformly in $\gamma \in M_1$ with probability at least $1-2\exp(-cn)$, by Lemma 1 in \cite{lee2021bayesian}.
	Hence, we have 
	\bean\label{gamma0_M1}
	f( \bfY_n \mid \bfX_{\gamma_0} , \gamma_0 ) \pi(\gamma_0 \mid G) 
	\,\gtrsim\, (\tau^2)^{-|\gamma_0|/2} \pi(\gamma_0 \mid G) \exp \big\{ L_n(\hat{\beta}_{\gamma_0} ) \big\}   \det \big(  (1+\epsilon) H_n( \beta_{0, \gamma_0} ) + \tau^{-2} I_{|\gamma_0|} \big)^{-1/2}   
	\eean
	for any $\gamma \in M_1$, with probability at least $1- 2 \exp(-cn)$.
	
	Then, \eqref{gamma_M1} and \eqref{gamma0_M1} implies 
	\bea
	\frac{\pi( \gamma , G \mid \bfY_n , \bfX_n )}{\pi( \gamma_0  , G \mid \bfY_n , \bfX_n )}  
	&\lesssim&  \frac{\pi(\gamma\mid G)}{\pi(\gamma_0 \mid G) } (n \tau^2)^{- \frac{1}{2}(|\gamma|- |\gamma_0|)  }  \frac{ \det \big( \frac{1+\epsilon}{n}H_n(\beta_{0,\gamma_0} )  + (n\tau^2)^{-1} I_{|\gamma_0|} \big)^{1/2}  }{ \det \big( \frac{1-\epsilon}{n}H_n(\beta_{0,\gamma} )  + (n\tau^2)^{-1} I_{|\gamma|} \big)^{1/2}  } \\
	&& \quad\quad\quad \times \,\, \exp \Big\{ L_n(\hat{\beta}_{\gamma})  - L_n(\hat{\beta}_{\gamma_0} )   \Big\}  \\
	&\lesssim&  
	\frac{\pi(\gamma\mid G)}{\pi(\gamma_0 \mid G) } (n \tau^2)^{- \frac{1}{2}(|\gamma|- |\gamma_0|)  }  \Big( \frac{2}{\lambda} \Big)^{|\gamma|-|\gamma_0| } 
	\exp \Big\{ L_n(\hat{\beta}_{\gamma})  - L_n(\hat{\beta}_{\gamma_0} )   \Big\}
	\eea
	for any $\gamma \in M_1$, with probability at least $1- 2 \exp(-cn)$, by Lemma 2 in \cite{lee2021bayesian}.
	
	Note that by Taylor's expansion of $L_n(\cdot)$,
	\bea
	L_n(\hat{\beta}_{\gamma})  - L_n(\hat{\beta}_{\gamma_0} ) 
	&\le& L_n(\hat{\beta}_{\gamma})  - L_n({\beta}_{0, \gamma} )  \\
	&=& (\hat{\beta}_\gamma - \beta_{0,\gamma} )^T s_n(\beta_{0,\gamma})  - \frac{1}{2} (\hat{\beta}_\gamma - \beta_{0,\gamma} )^T H_n( \tilde{\beta}_\gamma ) (\hat{\beta}_\gamma - \beta_{0,\gamma} )
	\eea
	for some $\tilde{\beta}_\gamma$ such that $\|\tilde{\beta}_\gamma - \hat{\beta}_\gamma\|_2 \le \| \beta_{0,\gamma} - \hat{\beta}_\gamma \|_2$.
	Again by Taylor's expansion of $s_n(\cdot)$, we have 
	$0 = s_n(\hat{\beta}_\gamma )  = s_n(\beta_{0,\gamma}) - H_n(\tilde{\beta}_\gamma^*) (\hat{\beta}_\gamma - \gamma_{0,\gamma} )$ for some $\tilde{\beta}_\gamma^*$ such that $\|\tilde{\beta}_\gamma^* - {\beta}_{0,\gamma}\|_2 \le \| \hat{\beta}_\gamma  - \beta_{0,\gamma}\|_2$, which implies
	\bea
	(\hat{\beta}_\gamma - \beta_{0,\gamma} )^T s_n(\beta_{0,\gamma})  
	&=& s_n(\beta_{0,\gamma} )^T  H_n( \tilde{\beta}_\gamma^* )^{-1}  s_n(\beta_{0,\gamma})
	\eea
	and
	\bea
	&&  (\hat{\beta}_\gamma - \beta_{0,\gamma} )^T H_n( \tilde{\beta}_\gamma ) (\hat{\beta}_\gamma - \beta_{0,\gamma} )  \\
	&=& s_n(\beta_{0,\gamma} )^T  H_n( \tilde{\beta}_\gamma^* )^{-1}  s_n(\beta_{0,\gamma})
	+
	(\hat{\beta}_\gamma - \beta_{0,\gamma} )^T \big\{H_n( \tilde{\beta}_\gamma ) - H_n( \tilde{\beta}_\gamma^* )  \big\} (\hat{\beta}_\gamma - \beta_{0,\gamma} )  .
	\eea
	Note that $H_n(\beta_\gamma) = \bfX_\gamma^T \Psi_\gamma \bfX_\gamma$ and 
	\bea
	&&  \Big| (\hat{\beta}_\gamma - \beta_{0,\gamma} )^T \big\{H_n( \tilde{\beta}_\gamma ) - H_n( \tilde{\beta}_\gamma^* )  \big\} (\hat{\beta}_\gamma - \beta_{0,\gamma} )  \Big|  \\
	&\le&  \sup_{u \in\bbR^{|\gamma|} : \|u\|_2=1} \Big| u^T \big\{H_n( \tilde{\beta}_\gamma ) - H_n( \tilde{\beta}_\gamma^* )  \big\}  u \Big| \cdot  \|\hat{\beta}_\gamma - \beta_{0,\gamma}\|_2^2 \\
	&\le&  \max_{1\le i\le n} \big|  \psi_i(\tilde{\beta}_\gamma)  -  \psi_i(\tilde{\beta}_\gamma^*)   \big| \cdot \|\bfX_\gamma^T \bfX_\gamma \|  \cdot  \|\hat{\beta}_\gamma - \beta_{0,\gamma}\|_2^2   \\
	&\lesssim&  \max_{1\le i\le n}  \big| X_{i,\gamma}^T \tilde{\beta}_\gamma - X_{i,\gamma}^T \tilde{\beta}_\gamma^* \big| \cdot \|\bfX_\gamma^T \bfX_\gamma \|  \cdot  \|\hat{\beta}_\gamma - \beta_{0,\gamma}\|_2^2    \\
	&\lesssim& R_2 \sqrt{\frac{\log p}{n} } \cdot n \cdot \frac{ R_2 \log p }{n}
	\,\,=\,\, \Big( \frac{ R_2^4 \log p }{n} \Big)^{1/2} \log p \,\,=\,\, o \big( \log p \big) 
	\eea
	uniformly in $\gamma \in M_1$ with probability at least $1- 2 \exp(-cn)$, 
	where the third inequality holds due to the Lipschitz continuity of $\psi_i({\beta}_\gamma)$ (using similar arguments in the proof of Lemma \ref{lem:hessian}) and the fourth inequality holds due to  condition (A2) and \eqref{eigenvals}. 
	Therefore, by Lemma \ref{lem:hessian},
	\bea
	L_n(\hat{\beta}_{\gamma})  - L_n(\hat{\beta}_{\gamma_0} ) 
	&\le& \frac{1}{2(1-\epsilon) }  s_n(\beta_{0,\gamma} )^T  H_n( {\beta}_{0,\gamma} )^{-1}  s_n(\beta_{0,\gamma})  + o (\log p)
	\eea
	uniformly in $\gamma \in M_1$ with probability at least $1- 2 \exp(-cn)$.
	Note that 
	$s_n(\beta_{0,\gamma}) = \bfX_\gamma^T D( \beta_{0,\gamma} ) \tilde{U}$, where 	$\tilde{U} = \sg^{-1/2} (Y - \mu)$, and 
	\bea
	H_n ( \beta_{0,\gamma} ) 
	&=& \sum_{i=1}^n X_{i,\gamma} X_{i,\gamma}^T \psi_i(\beta_{0,\gamma})  \\
	&\ge& \sum_{i=1}^n X_{i,\gamma} X_{i,\gamma}^T \cdot C^{-1} d_i( \beta_{0,\gamma} )^2
	\,\,\ge\,\, C^{-1} \bfX_\gamma^T D( \beta_{0,\gamma} )^2 \bfX_\gamma 
	\eea
	for some constant $C>0$, because $\psi_i(\beta_{0,\gamma})  \ge C^{-1} d_i ( \beta_{0,\gamma})^2$ on $|\eta_{0,i,\gamma}|  \le \|X_i\|_{\max} \| \beta_0\|_1 \le C'$ for some positive constants $C$ and $C'$.
	Let $P_\gamma = D(\beta_{0,\gamma}) \bfX_\gamma ( \bfX_\gamma^T D(\beta_{0,\gamma})^2 \bfX_\gamma )^{-1} \bfX_\gamma^T D(\beta_{0,\gamma})$.
	Thus,
	\bea
	L_n(\hat{\beta}_{\gamma})  - L_n(\hat{\beta}_{\gamma_0} ) 
	&\le& \frac{C}{2(1-\epsilon) }   \tilde{U}^T P_\gamma \tilde{U} + o \big( \log p \big)  \\
	&\le& C' ( |\gamma| - |\gamma_0| ) \log p
	\eea
	for some large positive constants $C$ and $C'$ uniformly in $\gamma \in M_1$ with probability at least $1-  2 \exp(-cn) - p^{-2|\gamma|}$ by Lemma \ref{lem:projection} with $t = 2|\gamma| \log p$ due to condition (A3).

	Then, we have 
	\bea
	\frac{\pi( \gamma , G \mid \bfY_n , \bfX_n )}{\pi( \gamma_0  , G \mid \bfY_n , \bfX_n )}  
	&\lesssim& 
	\frac{\pi(\gamma\mid G)}{\pi(\gamma_0 \mid G) } (n \tau^2)^{- \frac{1}{2}(|\gamma|- |\gamma_0|)  }  \Big( \frac{2}{\lambda} \Big)^{|\gamma|-|\gamma_0| } 
	\exp \Big\{ L_n(\hat{\beta}_{\gamma})  - L_n(\hat{\beta}_{\gamma_0} )   \Big\} \\
	&\lesssim& 
	\exp \big\{  - a ( |\gamma| - |\gamma_0| ) + b \gamma^T {G} \gamma - b \gamma_0^T {G} \gamma_0  \big\}
	\Big\{  p^{-(1+\delta - C )} \frac{2}{\lambda}  \Big\}^{ |\gamma| - |\gamma_0| }  \\
	&\le&  \Big\{  p^{-(1+\delta - C + C' )} \frac{2}{\lambda}  \Big\}^{ |\gamma| - |\gamma_0| }  
	\,\,=\,\, o(1)
	\eea
	uniformly in $\gamma \in M_1$ with probability at least $1-  2 \exp(-cn) - p^{-2|\gamma|}$ for some positive constants $C'$ and $\delta$ due to condition (A6).

	Now we focus on the remaining models, $M_2 = \big\{ \gamma  : \gamma \nsupseteq \gamma_0 , \, |\gamma| \le R_2  \big\}$.
	For any $\gamma \in M_2$, let $\gamma^* = \gamma \cup \gamma_0$  so that $\gamma^* \in M_1^* = \{ \gamma : \gamma \supset \gamma_0 , \, |\gamma| \le R_2 + |\gamma_0| \}$.
	Let $\beta_{\gamma^*}$ be the $|\gamma^*|$-dimensional vector including $\beta_\gamma$ for $\gamma$ and zeros for $\gamma_0\setminus \gamma$.
	By Taylor's expansion, 
	for any $\beta_{\gamma^*}$ such that $\|\beta_{\gamma^*}- \beta_{0,\gamma^*}\|_2 \le C \sqrt{ |\gamma^*| \log p /n} \equiv C w_n'$ for some large constant $C>0$, 
	\bea
	L_n( \beta_{\gamma^*} )
	&=& L_n( \hat{\beta}_{\gamma^*} )  - \frac{1}{2} (\beta_{\gamma^*} - \hat{\beta}_{\gamma^*} )^* H_n(\tilde{\beta}_{\gamma^*} )  (\beta_{\gamma^*} - \hat{\beta}_{\gamma^*} ) \\
	&\le& L_n( \hat{\beta}_{\gamma^*} )  - \frac{1-\epsilon}{2} (\beta_{\gamma^*} - \hat{\beta}_{\gamma^*} )^* H_n({\beta}_{0, \gamma^*} )  (\beta_{\gamma^*} - \hat{\beta}_{\gamma^*} )  \\
	&\le&  L_n( \hat{\beta}_{\gamma^*} )  - \frac{n ( 1-\epsilon)\lambda }{2} \| \beta_{\gamma^*} - \hat{\beta}_{\gamma^*}  \|_2^2
	\eea
	with probability at least $1-2\exp(-cn)$, 
	where the first inequality holds due to Lemmas \ref{lem:hessian} and \ref{lem:betahat}, 
	and the second inequality holds due to Lemma \ref{lem:hessian_lower}.
	Define $B_{\gamma^*} = \{ \beta_\gamma : \|\beta_{\gamma^*} - \hat{\beta}_{\gamma^*} \|_2 \le C w_n' /2 \}$, then $B_{\gamma^*} \subset \{ \beta_\gamma : \|\beta_{\gamma^*} - \beta_{0,{\gamma^*}} \|_2 < C w_n' \}$ with probability at least $1- 2\exp(-cn)$ uniformly in $\gamma \in M_2$.
	Then, for any $\gamma \in M_2$, with probability at least $1-2\exp(-cn)$, 
	\bea
	&&  f(\bfY_n \mid \bfX_{\gamma} , \gamma) \pi(\gamma \mid G )  \\
	&\le& \pi(\gamma \mid G)  \exp \big\{ L_n( \hat{\beta}_{\gamma^*} ) \big\} \bigg[  (2\pi \tau^2)^{- |\gamma|/2 } \int_{B_{\gamma^*}}  \exp \Big\{ -\frac{n(1-\epsilon)\lambda}{2} \|\beta_{\gamma^*} - \hat{\beta}_{\gamma^*} \|_2^2 - \frac{1}{2\tau^2} \|\beta_\gamma\|_2^2   \Big\} d \beta_\gamma  \\
	&& \quad\quad\quad\quad\quad\quad\quad\quad\quad\quad\quad\quad + \,\, (\tau^2)^{-|\gamma|/2}  \exp \big( - C' |\gamma^*| \log p \big)  \bigg]  \\
	&\le& \pi(\gamma \mid G)  \exp \big\{ L_n( \hat{\beta}_{\gamma^*} ) \big\} (\tau^2)^{-|\gamma|/2}  \bigg[ \big\{ n (1-\epsilon) \lambda + \tau^{-2} \big\}^{-|\gamma|/2} \exp \Big\{ - \frac{n(1-\epsilon)\lambda}{2} \| \hat{\beta}_{\gamma_0 \setminus \gamma} \|_2^2  \Big\}  \\
	&& \quad\quad\quad\quad\quad\quad\quad\quad\quad\quad\quad\quad + \,\,  \exp \big( - C' |\gamma^*| \log p \big)  \bigg] \\
	&\le& \pi(\gamma \mid G)  \exp \big\{ L_n( \hat{\beta}_{\gamma^*} ) \big\} (\tau^2)^{-|\gamma|/2}  
	\big\{ n (1-\epsilon) \lambda + \tau^{-2} \big\}^{-|\gamma|/2}   \\
	&& \quad\quad\quad\quad\quad\quad\quad\quad\quad\quad\quad\quad \times \exp \Big\{  - \frac{(1-\epsilon)\lambda}{2} \big( \frac{C_{\beta}}{2} |\gamma_0 \setminus \gamma| - C \big) |\gamma_0 \setminus \gamma| \log p   \Big\}  ( 1+ o(1))
	\eea
	for some positive constants $C$ and $C'$ because $|\gamma_0|=O(1)$,
	\bea
	&& (2\pi)^{-|\gamma|/2} \int_{B_{\gamma^*}} \exp \Big\{ -\frac{n(1-\epsilon)\lambda}{2} \|\beta_{\gamma^*} - \hat{\beta}_{\gamma^*} \|_2^2 - \frac{1}{2\tau^2} \|\beta_\gamma\|_2^2   \Big\} d \beta_\gamma  \\
	 &\le& \big\{ n (1-\epsilon) \lambda + \tau^{-2} \big\}^{-|\gamma|/2} \exp \Big\{ - \frac{n(1-\epsilon)\lambda}{2} \| \hat{\beta}_{\gamma_0 \setminus \gamma} \|_2^2  \Big\}
	\eea
	and
	\bea
	&& \exp \Big\{ - \frac{n(1-\epsilon)\lambda}{2} \| \hat{\beta}_{\gamma_0 \setminus \gamma} \|_2^2  \Big\} \\
	&\le& \exp \Big\{ - \frac{n(1-\epsilon)\lambda}{2} \Big( \frac{1}{2} \|\beta_{0, \gamma_0 \setminus \gamma} \|_2^2  - \| \hat{\beta}_{\gamma_0 \setminus \gamma}   - \beta_{0, \gamma_0 \setminus \gamma} \|_2^2  \Big) \Big\}  \\
	&\le&  \exp \Big\{ - \frac{n(1-\epsilon)\lambda}{2} \Big( \frac{1}{2} \|\beta_{0, \gamma_0 \setminus \gamma} \|_2^2  - C \frac{| \gamma_0 \setminus \gamma| \log p }{n}  \Big) \Big\} \\
	&\le&  \exp \Big\{ - \frac{(1-\epsilon)\lambda}{2} \Big( \frac{C_{\beta_0} }{2} |\gamma_0\setminus \gamma|  - C  \Big) |\gamma_0\setminus \gamma| \log p \Big\} 
	\eea
	due to condition (A3).

	By deriving the lower bound of $f(\bfY_n \mid \bfX_{\gamma_0}. \gamma_0)$ as before, 
	\bea
	\frac{\pi( \gamma , G \mid \bfY_n , \bfX_n )}{\pi( \gamma_0  , G \mid \bfY_n , \bfX_n )}  
	&\lesssim& 
	\frac{\pi(\gamma \mid G)}{\pi(\gamma_0 \mid G) } (n \tau^2)^{- \frac{1}{2}(|\gamma|- |\gamma_0| ) }  
	\frac{ \det \big(  \frac{1+\epsilon}{n} H_n( \beta_{0,\gamma_0}) + (n\tau^2)^{-1} I_{|\gamma_0|}   \big)^{1/2} }{ \{ (1-\epsilon)\lambda +  (n\tau^2)^{-1} \}^{|\gamma|/2} } \\
	&&   \times \,\, \exp \big\{  L_n( \hat{\beta}_{\gamma^*} ) - L_n( \hat{\beta}_{\gamma_0} ) \big\}  
	\exp \Big\{ - \frac{(1-\epsilon)\lambda}{2} \Big( \frac{C_{\beta_0} }{2} |\gamma_0\setminus \gamma|  - C  \Big) |\gamma_0\setminus \gamma| \log p \Big\}  \\
	&\lesssim& \frac{\pi(\gamma \mid G)}{\pi(\gamma_0 \mid G) } (\tilde{C} n \tau^2)^{- \frac{1}{2}(|\gamma|- |\gamma_0| ) }     \\
	&& \times \,\,  \exp \big\{  L_n( \hat{\beta}_{\gamma^*} ) - L_n( \hat{\beta}_{\gamma_0} ) \big\}  \exp \Big\{ - \frac{(1-\epsilon)\lambda}{2} \Big( \frac{C_{\beta_0} }{2} |\gamma_0\setminus \gamma|  - C  \Big) |\gamma_0\setminus \gamma| \log p \Big\} 
	\eea
	for any $\gamma \in M_2$ and some constant $\tilde{C}>0$ with probability at least $1- 2\exp(-cn) - p^{-2|\gamma^*|}$, because 
	\bea
	\frac{ \det \big(  \frac{1+\epsilon}{n} H_n( \beta_{0,\gamma_0}) + (n\tau^2)^{-1} I_{|\gamma_0|}   \big)^{1/2} }{ \{ (1-\epsilon)\lambda +  (n\tau^2)^{-1} \}^{|\gamma|/2} }  
	&\le& \frac{\{ (1+\epsilon)C_2^* +  (n\tau^2)^{-1} \}^{|\gamma_0|/2} }{\{ (1-\epsilon)\lambda +  (n\tau^2)^{-1} \}^{|\gamma|/2} }  \\
	&=& \frac{\{ (1+\epsilon)C_2^* +  (n\tau^2)^{-1} \}^{|\gamma_0|/2} }{\{ (1-\epsilon)\lambda +  (n\tau^2)^{-1} \}^{|\gamma_0|/2} }  
	\{ (1-\epsilon)\lambda +  (n\tau^2)^{-1} \}^{-\frac{1}{2}(|\gamma|-|\gamma_0|) }  \\
	&\lesssim& \tilde{C}^{-\frac{1}{2}(|\gamma|-|\gamma_0|) } .
	\eea
	Therefore, by similar arguments used for $\gamma \in M_1$ case, 
	\bea
	\frac{\pi( \gamma , G \mid \bfY_n , \bfX_n )}{\pi( \gamma_0  , G \mid \bfY_n , \bfX_n )}  
	&\lesssim& 
	\exp \Big\{  - \big( C_a + 1+ \delta + \frac{\log \tilde{C} }{\log p}  \big) (|\gamma|-|\gamma_0| )\log p      + C' (|\gamma^*|- |\gamma|) \log p  \\
	&& \quad\quad - \,\, \frac{(1-\epsilon)\lambda}{2} \Big( \frac{C_{\beta_0} }{2} |\gamma_0\setminus \gamma|  - C  \Big) |\gamma_0\setminus \gamma| \log p    \Big\}   \\
	&=& \exp \Big[  - \big( C_a + 1+ \delta + \frac{\log \tilde{C} }{\log p} - C'  \big) (|\gamma| - |\gamma \cap \gamma_0| ) \log p\\
	&& \quad\quad - \Big\{ \frac{(1-\epsilon)\lambda}{2}\Big( \frac{C_{\beta_0} }{2} |\gamma_0\setminus \gamma|  - C  \Big) - C_a -1 -\delta   \Big\} |\gamma_0 \setminus \gamma| \log p   \Big] \\
	&=& o(1)
	\eea
	uniformly in $\gamma \in M_2$ with probability at least $1- 2\exp(-cn) - p^{-2|\gamma^*|}$ for some positive constants $C_a, C_{\beta_0}$ and $\delta$, which completes the proof. 	
\end{proof}

\begin{proof}[Proof of Theorem \ref{thm:G_cons}]
	By Theorem 1 in \cite{jalali2020bconcord}, under conditions (A2), (A4) and (A5), 
	\bea
	\pi(G_0 \mid \hat{\delta}, \bfX_n)  &\overset{P}{\lra}  1  \quad \text{ as } n\to\infty .
	\eea
	It implies
	\bea
	\frac{\pi(\gamma_0 ,G \mid \hat{\delta},  \bfY_n, \bfX_n) }{\pi(\gamma_0 ,G_0  \mid \hat{\delta}, \bfY_n, \bfX_n)}
	&=& \frac{f(\bfX_n \mid \hat{\delta}, G)  \pi(\gamma_0 \mid G) \pi(G) }{ f(\bfX_n \mid \hat{\delta}, G_0)  \pi(\gamma_0 \mid G_0) \pi(G_0)  }  \\
	&=& \frac{ \pi( G \mid \hat{\delta}, \bfX_n) }{ \pi( G_0 \mid \hat{\delta}, \bfX_n) }  \frac{\pi(\gamma_0 \mid G)}{\pi(\gamma_0 \mid G_0)}    \\
	&=& \frac{ \pi( G \mid \hat{\delta}, \bfX_n) }{ \pi( G_0 \mid \hat{\delta}, \bfX_n) } \exp \big( b \gamma_0^T {G} \gamma_0 - b \gamma_0^T {G}_0 \gamma_0  \big) \\
	&\le&  \frac{ \pi( G \mid \hat{\delta}, \bfX_n) }{ \pi( G_0 \mid \hat{\delta}, \bfX_n) } \exp \big( b |\gamma_0|^2 \big) 
	\,\,\overset{P}{\lra}\,\, 0 \quad \text{ as } n\to\infty  
	\eea
	for any $G \neq G_0$, where $\pi( \bfX_n \mid  \hat{\delta}, G ) = \int \pi(\bfX_n \mid \xi,  \hat{\delta}, G) \pi(\xi \mid G)  d \xi$.
\end{proof}

\begin{proof}[Proof of Corollary \ref{thm:joint_cons}]
	Note that 
	\bea
	\frac{\pi(\gamma ,G \mid \hat{\delta}, \bfY_n, \bfX_n) }{\pi(\gamma_0 ,G_0 \mid \hat{\delta}, \bfY_n, \bfX_n)}
	&=& 
	\frac{\pi(\gamma ,G \mid \hat{\delta},  \bfY_n, \bfX_n) }{\pi(\gamma_0 ,G  \mid \hat{\delta}, \bfY_n, \bfX_n)}
	\frac{\pi(\gamma_0 ,G \mid \hat{\delta},  \bfY_n, \bfX_n) }{\pi(\gamma_0 ,G_0  \mid \hat{\delta}, \bfY_n, \bfX_n)} \\
	&=& 
	\frac{ f(\bfY_n \mid \bfX_{\gamma}, \gamma) \pi(\gamma \mid G) }{ f(\bfY_n \mid \bfX_{\gamma_0} , \gamma_0)  \pi(\gamma_0 \mid G) }
	\frac{\pi(\gamma_0 ,G \mid \hat{\delta},  \bfY_n, \bfX_n) }{\pi(\gamma_0 ,G_0  \mid \hat{\delta}, \bfY_n, \bfX_n)}    \\
	&=& 
	\frac{\pi(\gamma ,G \mid  \bfY_n, \bfX_n) }{\pi(\gamma_0 ,G  \mid \bfY_n, \bfX_n)}
	\frac{\pi(\gamma_0 ,G \mid \hat{\delta},  \bfY_n, \bfX_n) }{\pi(\gamma_0 ,G_0  \mid \hat{\delta}, \bfY_n, \bfX_n)} .
	\eea
	Thus, by applying Theorems \ref{thm:gamma_cons} and \ref{thm:G_cons}, we can complete the proof.
\end{proof}

\begin{proof}[Proof of Theorem \ref{thm:joint_sel}]
	It suffices to show that 
	\bea
	\sum_{\gamma: \gamma \neq \gamma_0} \sum_{G: G \neq G_0}   \frac{\pi(\gamma ,G \mid \hat{\delta}, \bfY_n, \bfX_n) }{\pi(\gamma_0 ,G_0 \mid \hat{\delta}, \bfY_n, \bfX_n)}
	&\overset{P}{\lra}& 0 \quad \text{ as } n \to\infty  .
	\eea
	Note that 
	\bea
	\sum_{\gamma: \gamma \neq \gamma_0} \sum_{G: G \neq G_0}  \frac{\pi(\gamma ,G \mid \hat{\delta}, \bfY_n, \bfX_n) }{\pi(\gamma_0 ,G_0 \mid \hat{\delta}, \bfY_n, \bfX_n)}
	&=& 
	\sum_{G: G \neq G_0}  \bigg\{   \frac{\pi(\gamma_0 ,G \mid \hat{\delta},  \bfY_n, \bfX_n) }{\pi(\gamma_0 ,G_0  \mid \hat{\delta}, \bfY_n, \bfX_n)}  \sum_{\gamma: \gamma \neq \gamma_0} \frac{\pi(\gamma ,G \mid  \bfY_n, \bfX_n) }{\pi(\gamma_0 ,G  \mid \bfY_n, \bfX_n)}   \bigg\}    
	\eea
	and $\{\gamma: \gamma \neq \gamma_0 , |\gamma| \le R_2 \} = M_1 \cup M_2$, where $M_1$ and $M_2$ are defined in the proof of Theorem \ref{thm:gamma_cons}.
	
	By the proof of Theorem \ref{thm:gamma_cons}, we have 
	\bea
	\sum_{\gamma \in M_1 } \frac{\pi(\gamma ,G \mid  \bfY_n, \bfX_n) }{\pi(\gamma_0 ,G  \mid \bfY_n, \bfX_n)}  
	&\lesssim&  
	\sum_{k= |\gamma_0|+1}^{R_2} \sum_{\gamma\in M_1: |\gamma|=k }   \Big\{  p^{-(1+\delta - C + C' )} \frac{2}{\lambda}  \Big\}^{ k - |\gamma_0| }   \\
	&\le&  \sum_{k= |\gamma_0|+1}^{R_2} \binom{p - |\gamma_0|}{k - |\gamma_0|}  \Big\{  p^{-(1+\delta - C + C' )} \frac{2}{\lambda}  \Big\}^{ k - |\gamma_0| }  \\
	&\le& \sum_{k= |\gamma_0|+1}^{R_2}    \Big\{  p^{-(\delta - C + C' )} \frac{2}{\lambda}  \Big\}^{ k - |\gamma_0| }   \,\,=\,\, o(1)
	\eea 
	for some positive constants $C'$ and $\delta$, 
	with probability at least $1-  2 \exp(-cn) - p^{- |\gamma_0|-1 }$, 
	and	
	\bea
	&&  \sum_{\gamma \in M_2 } \frac{\pi(\gamma ,G \mid  \bfY_n, \bfX_n) }{\pi(\gamma_0 ,G  \mid \bfY_n, \bfX_n)} \\
	&\lesssim& 
	\sum_{k =0}^{R_2}  \sum_{\gamma\in M_2: |\gamma|=k}   
	\exp \Big[  - \big( C_a + 1+ \delta + \frac{\log \tilde{C} }{\log p} - C'  \big) (k - |\gamma \cap \gamma_0| ) \log p\\
	&& \quad\quad - \Big\{ \frac{(1-\epsilon)\lambda}{2}\Big( \frac{C_{\beta_0} }{2} |\gamma_0\setminus \gamma|  - C  \Big) - C_a -1 -\delta   \Big\} |\gamma_0 \setminus \gamma| \log p   \Big]  \\
	&\le& 
	\sum_{k =0}^{R_2}  \sum_{\nu=0}^{ (|\gamma_0|-1)\wedge k }  \binom{|\gamma_0|}{\nu}  \binom{p- |\gamma_0|}{k - \nu}   \exp \Big[  - \big( C_a + 1+ \delta + \frac{\log \tilde{C} }{\log p} - C'  \big) ( k -\nu ) \log p\\
	&& \quad\quad - \Big\{ \frac{(1-\epsilon)\lambda}{2}\Big( \frac{C_{\beta_0} }{2} (|\gamma_0| - \nu)  - C  \Big) - C_a -1 -\delta   \Big\} (|\gamma_0| - \nu) \log p   \Big] \\
	&\le&  \sum_{k =0}^{R_2}  \sum_{\nu=0}^{ (|\gamma_0|-1)\wedge k }  |\gamma_0|^{|\gamma_0| - \nu}   \exp \Big[  - \big( C_a + 1+ \delta + \frac{\log \tilde{C} }{\log p} - C'   -1 \big) ( k -\nu ) \log p\\
	&& \quad\quad - \Big\{ \frac{(1-\epsilon)\lambda}{2}\Big( \frac{C_{\beta_0} }{2} (|\gamma_0| - \nu)  - C  \Big) - C_a -1 -\delta   \Big\} (|\gamma_0| - \nu) \log p   \Big] \\
	&=&  o(1)
	\eea
	for some positive constants $C_a, C_{\beta_0}$ and $\delta$, 
	with probability at least $1-  2 \exp(-cn) - p^{-2}$.

	On the other hand, we have 
	\bea
	\sum_{G: G \neq G_0} \frac{\pi(\gamma_0 ,G \mid \hat{\delta},  \bfY_n, \bfX_n) }{\pi(\gamma_0 ,G_0  \mid \hat{\delta}, \bfY_n, \bfX_n)}
	&\le& 
	\sum_{G: G \neq G_0} \frac{ \pi( G \mid \hat{\delta}, \bfX_n) }{ \pi( G_0 \mid \hat{\delta}, \bfX_n) } \exp \big( b |\gamma_0|^2 \big) \\
	&=& \frac{1 -  \pi( G_0 \mid \hat{\delta}, \bfX_n)}{  \pi( G_0 \mid \hat{\delta}, \bfX_n)}  \exp \big( b |\gamma_0|^2 \big) 
	\,\,\overset{P}{\lra}\,\, 0
	\eea
	by Theorem 1 in \cite{jalali2020bconcord}, which completes the proof.	
\end{proof}

\begin{lemma}\label{lem:hessian}
	Under conditions (A1)--(A3) and $R_2 = (n / \log p)^{\frac{1-d}{2} }$ with $1/2 < d < 1$, 
	for $\epsilon_n = C R_2 \sqrt{ \log p /n } $, 
	we have
	\bea
	(1-\epsilon_n)  H_n\big( \beta_{0,\gamma} \big) \,\,\le\,\,  H_n\big( \beta_{\gamma} \big) \,\,\le\,\, (1+\epsilon_n) H_n\big( \beta_{0,\gamma} \big)
	\eea
	for any $\gamma \in M_1^* = \{ \gamma : \gamma \supset \gamma_0 , \, |\gamma| \le R_2  + |\gamma_0| \}$ and any $\beta_\gamma$ such that $\|\beta_\gamma - \beta_{0,\gamma} \|_2 \le  \big(  C' |\gamma|  \log p /n \big)^{\frac{1}{2} }$, with probability at least $1 - 2 \exp ( - cn)$ for some positive constants $c, C$ and $C'$.
\end{lemma}

\begin{proof}[Proof of Lemma \ref{lem:hessian}]
	
	Note that by Theorem 5.39 and Remark 5.40 in \cite{eldar2012compressed}, there exist positive constants $C_1^*$ and $C_2^*$ such that 
	\bean\label{eigenvals}
	C_1^* \,\,\le\,\, \min_{\gamma: |\gamma| \le R_2 + |\gamma_0|} \lambda_{\min} \Big( n^{-1} \bfX_{\gamma}^T \bfX_{\gamma} \Big) 
	\,\,\le\,\,   \max_{\gamma: |\gamma| \le R_2 + |\gamma_0|} \lambda_{\max} \Big( n^{-1} \bfX_{\gamma}^T \bfX_{\gamma} \Big) 
	\,\,\le\,\, C_2^* 
	\eean
	with probability at least $1- 2 \exp(-cn)$ for some constant $c>0$.
	Also note that $H_n(\beta_\gamma)   = \bfX_\gamma^T  \Psi_\gamma  \bfX_\gamma$, where $\Psi_\gamma  = diag( \psi_i(\beta_\gamma) ) \in \bbR^{n\times n}$ and 
	\bean\label{psi_element}
	\psi_i(\beta_\gamma) 
	&=&  Y_i \Big\{  \frac{\eta_{i,\gamma} \phi(\eta_{i,\gamma})}{\Phi(\eta_{i,\gamma}) } +    \frac{\phi(\eta_{i,\gamma})^2}{\Phi(\eta_{i,\gamma})^2 }    \Big\}    + (1-Y_i) \Big\{  \frac{-\eta_{i,\gamma} \phi(\eta_{i,\gamma})}{1-\Phi(\eta_{i,\gamma}) } +    \frac{\phi(\eta_{i,\gamma})^2}{ (1- \Phi(\eta_{i,\gamma}))^2 }  \Big\}  .
	\eean
	Thus, it suffices to show that 
	\bea
	(1-\epsilon_n) \psi_i(\beta_{0,\gamma}) 
	\,\,\le\,\, \psi_i(\beta_\gamma) 
	\,\,\le\,\,
	(1+\epsilon_n) \psi_i(\beta_{0,\gamma}) , \quad \forall i =1,\ldots, n.
	\eea
	
	Due to conditions (A1)--(A3), uniformly for $\gamma \in M_1^*$,
	\bea
	| \eta_{i,\gamma} - \eta_{i,0,\gamma} | 
	&=& | X_{i, \gamma}^T \beta_\gamma - X_{i,\gamma}^T \beta_{0,\gamma} |  \,\,\le\,\, \|X_{i,\gamma} \|_2  \| \beta_\gamma - \beta_{0,\gamma}  \|_2 \\
	&\le& \sqrt{|\gamma|} M \,  \big(  C |\gamma|  \log p /n \big)^{\frac{1}{2} }  \,\,\lesssim\,\,  R_2 \big(  \log p /n \big)^{\frac{1}{2} }  \,\,=\,\, o (1) 
	\eea
	and
	\bea
	| \eta_{i,\gamma} | 
	&\le& | \eta_{i,\gamma} - \eta_{i,0,\gamma} |  + | \eta_{i,0,\gamma} |  \\
	&\le& o (1)  + \|X_i\|_{\max} \|\beta_0\|_1 \,\,\le\,\, C
	\eea
	for some constant $C>0$, with probability at least $1- 2 \exp(-cn)$.
	Since the support of $\eta_{i,\gamma}$, say $\calS$, is compact and $h( \eta_{i,\gamma}  ) \equiv \log \psi_i(\beta_\gamma)$ is continuously differentiable on $\calS$, $h( \eta_{i,\gamma}  ) $ is Lipschitz continuous with probability at least $1- 2 \exp(-cn)$.
	Therefore, for any $\gamma \in M_1^*$, 
	\bea
	\frac{\psi_i(\beta_\gamma)}{\psi_i(\beta_{0,\gamma})} &=& \exp \big\{ h( \eta_{i,\gamma}  )  - h( \eta_{i,0,\gamma}  )    \big\} \,\,\le\,\,  \exp \big\{ K |\eta_{i,\gamma}    -  \eta_{i,0,\gamma} | \big\}  \\
	&\le& \exp \Big\{  K \, C \, R_2 \sqrt{ \frac{\log p}{n} }   \Big\} 
	\,\,\le\,\, 1 + C'  R_2 \sqrt{ \frac{\log p}{n} } 
	\eea
	and
	\bea
	\frac{\phi_i(\beta_\gamma)}{\phi_i(\beta_{0,\gamma})} &\ge& 1 - C''  R_2 \sqrt{ \frac{\log p}{n} } 
	\eea
	for some positive constants $K$, $C$, $C'$ and $C''$,  with probability at least $1- 2 \exp(-cn)$.
	By taking $\epsilon_n = C R_2 \sqrt{ \log p /n } $, it completes the proof.
\end{proof}

\begin{lemma}\label{lem:projection}
	Let $\tilde{U} = \sg^{-1/2} (Y- \mu)$ and $P_\gamma$ be the projection matrix onto the column space of $D_{0,\gamma} \bfX_{\gamma}$, where $|\gamma| \le R_2$, and $R_2 = (n / \log p)^{\frac{1-d}{2} }$ with $1/2 < d < 1$.
	Then, for some constant $\delta^*>0$, 
	\bea
	\bbP_0 \Big[ \tilde{U}^T P_\gamma \tilde{U}  > (1+\delta^*) \big\{ {\rm tr}(P_\gamma) + 2 \sqrt{{\rm tr}(P_\gamma \, ) t }  + 2 t   \big\}    \Big]   &\le& e^{-t} , \quad \forall t > 0.
	\eea
\end{lemma}

\begin{proof}[Proof of Lemma \ref{lem:projection}]
	Because the distribution of $Y_i$ given $X_i$ is a Bernoulli distribution, there exists a  constant $\delta^*>0$ and $N(\delta^*)$ such that, given $\bfX_n$,
	\bea
	\bbE_0 \big\{  \exp( u^T \tilde{U})  \mid \bfX_n \big\} &\le& \exp \Big\{  \frac{1+\delta^*}{2} \|u\|_2^2  \Big\}
	\eea
	for any $n\ge N(\delta^*)$ and $u\in \bbR^n$ in the space spanned by the columns of $D_{0,\gamma} \bfX_{\gamma}$ (See condition 2(c) and related explanations in \cite{narisetty2019}).
	By Theorem A.1 in \cite{narisetty2019}, 
	\bea
	\bbP_0 \Big[ \tilde{U}^T P_\gamma \tilde{U}  > (1+\delta^*) \big\{ {\rm tr}(P_\gamma) + 2 \sqrt{{\rm tr}(P_\gamma \, ) t }  + 2 t   \big\}  \mid \bfX_n  \Big]   &\le& e^{-t} , \quad \forall t > 0,
	\eea
	which implies the desired result because the right-hand-side does not depend on $\bfX_n$.
\end{proof}

\begin{lemma}\label{lem:betahat}
	Under conditions (A1)--(A3) and $R_2 = (n / \log p)^{\frac{1-d}{2} }$ with $1/2 < d < 1$, we have 
	\bea
	\sup_{\gamma: \gamma \supset \gamma_0 , |\gamma|=m  } \| \hat{\beta}_{\gamma}  - \beta_{0,\gamma} \|_2 
	&=& O \bigg( \sqrt{\frac{m \log p}{n} } \, \bigg) 
	\eea
	uniformly for all $m \le R_2  + |\gamma_0|$ with probability at least $1- 2 \exp(-cn)$ for some constant $c>0$. 
\end{lemma}

\begin{proof}[Proof of Lemma \ref{lem:betahat}]
	Fix $\gamma$ such that $\gamma \supset \gamma_0$ and $|\gamma|=m$.
	Note that if $|\gamma| = |\gamma_0|=0$, then the above argument trivially holds.
	Thus, we focus only on the case $|\gamma_0| \ge 1$.
	Let $\tilde{Y}=( \tilde{Y}_i ) \in \bbR^n$ and $\tilde{\mu}= (\tilde{\mu}_i) \in \bbR^n$, where $\eta_{i,0} = X_i^T \beta_0$, $\tilde{Y}_{i}  = \phi (\eta_{i,0}) Y_i / \{ \Phi(\eta_{i,0})  ( 1- \Phi(\eta_{i,0})  ) \}$ and $\tilde{\mu}_i =  \phi (\eta_{i,0})  \Phi(\eta_{i,0}) / \{ \Phi(\eta_{i,0})  ( 1- \Phi(\eta_{i,0})  ) \}$.
	Then, $s_n(\beta_{0,\gamma})  = \bfX_{\gamma}^T (\tilde{Y}- \tilde{\mu} )$.
	Note that for any $\tilde{\alpha}\in \bbR^n$, 
	\bea
	\bbE_0 \exp \big\{ \tilde{\alpha}^T ( \tilde{Y}- \tilde{\mu})  \big\} 
	&=&  \bbE_0 \Big( \bbE_0\Big[ \exp \big\{ \tilde{\alpha}^T ( \tilde{Y}- \tilde{\mu})  \big\}  \mid \bfX_n \Big] \Big)  \\
	&\equiv& \bbE_0 \Big( \bbE_0\Big[ \exp \big\{ {\alpha}^T ( {Y}- {\mu})  \big\}  \mid \bfX_n \Big] \Big)  \\
	&\le& \bbE_0 \Big\{   \exp \Big( \frac{1}{8} \sum_{i=1}^n \alpha_i^2 \Big)  \Big\}  \\
	&=& \bbE_0 \Big\{   \exp \Big( \frac{1}{8} \sum_{i=1}^n \tilde{\alpha}_i^2 \frac{\phi (\eta_{i,0})^2}{ \Phi(\eta_{i,0})^2  ( 1- \Phi(\eta_{i,0})  )^2 } \Big)  \Big\}  \\
	&\le& \bbE_0 \Big\{   \exp \Big( \frac{\phi (M_0)^2}{ 8 \Phi( M_0)^2  ( 1- \Phi( M_0 )  )^2 }  \sum_{i=1}^n \tilde{\alpha}_i^2 \Big)  \Big\}
	\,\,\equiv \,\, \bbE_0 \Big\{   \exp \Big( \frac{1}{2} \sigma^2\| \tilde{\alpha}\|_2^2  \Big)  \Big\}
	\eea
	where $\alpha = ( \alpha_i )\in\bbR^n$ and $\alpha_i = \tilde{\alpha}_i  \phi (\eta_{i,0})  / \{ \Phi(\eta_{i,0})  ( 1- \Phi(\eta_{i,0})  ) \}$.
	The last inequality holds because $|\eta_{i,0}| \le \| X_i\|_{\max} \|\beta_0\|_1 \le M_0$ for some constant $M_0$.
	Therefore, by Theorem A.1 in \cite{narisetty2019}, it implies 
	\bea
	p^{-2m}  &\ge& 
	\bbP_0 \Big\{  \| \bfX_{\gamma}^T  ( \tilde{Y}- \tilde{\mu})  \|_2^2  \ge \sigma^2 \big( {\rm tr}(\bfX_{\gamma} \bfX_{\gamma}^T) + 2 \sqrt{{\rm tr}(\bfX_{\gamma} \bfX_{\gamma}^T\bfX_{\gamma} \bfX_{\gamma}^T) \, 2 m \log p } + 4 \|\bfX_{\gamma} \bfX_{\gamma}^T\| m \log p  \big)  \mid \bfX_n \Big\}  \\
	&\ge& \bbP_0 \Big\{  \| \bfX_{\gamma}^T  ( \tilde{Y}- \tilde{\mu})  \|_2^2  \ge \sigma^2 \big( {\rm tr}(\bfX_{\gamma}^T \bfX_{\gamma} ) + 2 {\rm tr}\{(\bfX_{\gamma}^T \bfX_{\gamma} )^2\} \sqrt{  2 m \log p } + 4 \|\bfX_{\gamma}^T \bfX_{\gamma}\| m \log p  \big) \mid \bfX_n  \Big\}   \\
	&\ge& \bbP_0 \Big\{  \| \bfX_{\gamma}^T  ( \tilde{Y}- \tilde{\mu})  \|_2^2  \ge 
	\sigma^2 \big( C_2^* m n  + 2\sqrt{2} C_2^* m n \sqrt{  \log p } + 4 C_2^* m n \log p  \big)  \mid \bfX_n  \Big\}   \\
	&\ge& \bbP_0 \Big\{  \| \bfX_{\gamma}^T  ( \tilde{Y}- \tilde{\mu})  \|_2^2  \ge 
	5 C_2^* \sigma^2  m n \log p   \mid \bfX_n  \Big\} 
	\eea
	for all sufficiently large $p$, with probability at least $1- 2 \exp(-cn)$ for some constant $c>0$.
	Here, $C_2^*$ is defined at \eqref{eigenvals}.
	
	Now, let $\beta_\gamma  = \beta_{0,\gamma} + c_n u $, where $u \in \bbR^n$, $\|u\|_2= 1$, 
	$c_n = \sqrt{20 C_2^* \sigma^2  m \log p / \{ n \lambda^2 ( 1-\epsilon)^2  \} }$ for some small constant $\epsilon>0$ and $C_1^*$ is defined at \eqref{eigenvals}.
	Then,
	\bea
	&& \bbP_0 \Big\{  L_n(\beta_\gamma )  - L_n(\beta_{0,\gamma})   > 0 \,\, \text{ for some } u  \mid \bfX_n  \Big\}  \\
	&=& \bbP_0 \Big\{ c_n u^T s_n( \beta_{0,\gamma} )  > \frac{1}{2} c_n^2 u^T H_n ( \tilde{\beta}_\gamma ) u    \,\, \text{ for some } u   \mid \bfX_n\Big\}  \\
	&\le&  \bbP_0 \Big\{  u^T s_n( \beta_{0,\gamma} )  > \frac{1}{2} (1-\epsilon) c_n    \lambda_{\min} \big( H_n( \beta_{0,\gamma} ) \big)    \,\, \text{ for some } u   \mid \bfX_n\Big\}  \\
	&\le&  \bbP_0 \Big\{  u^T s_n( \beta_{0,\gamma} )  > \frac{1}{2} (1-\epsilon) \lambda c_n   n  \,\, \text{ for some } u   \mid \bfX_n \Big\}  \\
	&\le& \bbP_0 \Big\{ \| \bfX_{\gamma}^T (\tilde{Y}- \tilde{\mu} )  \|_2^2  >  \frac{(1-\epsilon)^2}{4}  \lambda^2 c_n^2   n^2     \mid \bfX_n \Big\}  \\
	&=& \bbP_0 \Big\{ \| \bfX_{\gamma}^T (\tilde{Y}- \tilde{\mu} )  \|_2^2  >  5 C_2^*  \sigma^2 m n \log p     \mid \bfX_n \Big\}  
	\,\,\le\,\, p^{-2m} 
	\eea
	with probability at least $1- 2 \exp(-cn)$ for some constant $c>0$,
	where the first equality holds due to the Taylor's expansion,
	the first inequality holds due to Lemma \ref{lem:hessian},
	the second inequality holds due to Lemma \ref{lem:hessian_lower}.
	Due to the concavity of $L_n(\cdot)$, it implies that 
	\bea
	\bbP_0 \Big(  \sup_{\gamma: \gamma \supset \gamma_0 , |\gamma|=m  } \| \hat{\beta}_{\gamma}  - \beta_{0,\gamma} \|_2   > c_n  \,\, \text{ for any } m \le R_2 \Big)  
	&\le& \sum_{m=|\gamma_0|}^{ R_2+ |\gamma_0|} \binom{p}{m} p^{-2m}  + 2e^{-cn}
	\,\,=\,\, o (1) ,
	\eea
	which gives the desired result.
\end{proof}

\begin{lemma}\label{lem:hessian_lower}
	Under conditions(A1)--(A3) and $R_2 = (n / \log p)^{\frac{1-d}{2} }$ with $1/2 < d < 1$, there exists a constant $\lambda>0$ such that 
	\bea
	\lambda &\le&  \min_{\gamma : |\gamma| \le R_2 + |\gamma_0| } \lambda_{\min} \Big( \frac{1}{n} H_n( \beta_{0,\gamma} ) \Big) 
	\eea
	with probability at least $1- 2 \exp(-cn)$ for some constant $c>0$.
\end{lemma}

\begin{proof}[Proof of Lemma \ref{lem:hessian_lower}]
	Let $[n] =\{1,\ldots, n\}$ and $I = \{ i \in [n]:  |X_i^T \beta_0 | \le M \}$ for some large constant $M>0$ satisfying, for some constant $0<w<1$,
	\bean\label{bound_xbeta}
	\max_{1\le i \le n}\bbP_0 \big( | X_i^T \beta_0 | \ge M \big) &\le& w \,\,<\,\, 1.
	\eean
	Note that a constant $M$ satisfying \eqref{bound_xbeta} always exists due to conditions (A2) and (A3).
	Then, by Hoeffding's inequality,  $\bbP_0 \big\{  |I|  \ge n(1-w)/2  \big\} \ge 1 - \exp \{ -n (1-w) \}$. 
	Furthermore, by \eqref{eigenvals}, for any $\gamma$ such that $|\gamma| \le R_2 + |\gamma_0|$,
	\bea
	\lambda_{\min} \Big( \frac{1}{n} H_n( \beta_{0,\gamma} ) \Big) 
	&=& \lambda_{\min} \Big( \frac{1}{n}  \bfX_\gamma^T   diag( \psi_i(\beta_{0,\gamma}) )  \bfX_\gamma \Big)  \\
	&\ge& \lambda_{\min} \Big( \frac{1}{n}  \bfX_{I,\gamma}^T   diag( \psi_i(\beta_{0,\gamma}) )  \bfX_{I,\gamma} \Big) \\
	&\ge& \frac{1}{2} C_1^*(1-w) \, \min_{i\in I} \psi_i(\beta_{0,\gamma} ) 
	\eea
	with probability at least $1- 2 \exp(-cn)$ for some constant $c>0$, 
	where the definition of $ \psi_i(\beta_{0,\gamma})$ is given at \eqref{psi_element}
	and $\bfX_{I,\gamma} \in \bbR^{|I| \times |\gamma| }$ is a submatrix of $\bfX_{\gamma}$ consisting of the columns corresponding to $I$.
	Note that for any $i\in I$, 
	\bea
	\psi_i(\beta_{0,\gamma} ) 
	&\ge& M \frac{ \phi(M) }{\Phi(M)}  +  \frac{ \phi(M)^2 }{\Phi(M)^2 } \,\,\equiv\,\, d_M .
	\eea
	Therefore, 
	\bea
	\min_{\gamma : |\gamma| \le R_2 + |\gamma_0| } \lambda_{\min} \Big( \frac{1}{n} H_n( \beta_{0,\gamma} ) \Big) 
	&\ge& \frac{1}{2} d_M  C_1^*(1-w) 
	\eea
	with probability at least $1- 2 \exp(-cn)$ for some constant $c>0$.
\end{proof}

\bibliographystyle{plainnat}
\bibliography{references}
\end{document}